\documentclass[11pt]{amsart}
\usepackage{amsmath,mathtools}
\usepackage{amsmath, amsthm, amssymb, amsfonts, enumerate}
\usepackage[colorlinks=true,linkcolor=blue,urlcolor=blue]{hyperref}
\usepackage{dsfont}
\usepackage{color}
\usepackage{geometry}
\usepackage{todonotes}
\usepackage{epstopdf}
\usepackage{bbm}
\usepackage{geometry}
\usepackage[utf8]{inputenc}

\usepackage{amsfonts}
\usepackage{amsfonts}
\usepackage{textcomp}
\usepackage{amssymb}
\usepackage{float}
\usepackage{tikz}
\usepackage{epsfig}
\usepackage{amsmath}
\usepackage[english]{babel}
\usepackage{a4}
\usepackage{enumerate}
\usepackage{amsmath}
\newcommand{\stkout}[1]{\ifmmode\text{\sout{\ensuremath{#1}}}\else\sout{#1}\fi}

\newtheorem{theorem}{Theorem}[section]
\newtheorem{remark}[theorem]{Remark}
\newtheorem{assumption}[theorem]{Assumption}
\newtheorem{lemma}[theorem]{Lemma}
\newtheorem{proposition}[theorem]{Proposition}
\newtheorem{corollary}[theorem]{Corollary}

\newcommand{\one}{\text{$\mathbbm{1}$}}

\title[Irreversible Investment Problem with Two-Factor Uncertainty]{On an Irreversible Investment Problem with Two-Factor Uncertainty}

\author[Dammann]{F. Dammann\dag}
\author[Ferrari]{G. Ferrari\dag}

\date{\today \\\text{} \\ Felix Dammann \\ \href{mailto:}{dammann@uni-bielefeld.de} \\ Giorgio Ferrari \\ \href{mailto:}{ giorgio.ferrari@uni-bielefeld.de} \\ \text{} \\
\dag Center for Mathematical Economics (IMW), Bielefeld University, Universit\"{a}tsstraße 25, D-33615 Bielefeld, Germany 
}

\numberwithin{equation}{section}

\begin{document}

\maketitle

\begin{abstract}
We consider a real options model for the optimal irreversible investment problem of a profit maximizing company. The company has the opportunity to invest into a production plant capable of producing two products, of which the prices follow two independent geometric Brownian motions. After paying a constant sunk investment cost, the company sells the products on the market and thus receives a continuous stochastic revenue-flow. This investment problem is set as a two-dimensional optimal stopping problem. We find that the optimal investment decision is triggered by a convex curve, which we characterize as the unique continuous solution to a nonlinear integral equation. Furthermore, we provide analytical and numerical comparative statics results of the dependency of the project's value and investment decision with respect to the model's parameters.  
\end{abstract}
\text{}
\section{Introduction}
In this paper, we study a real options model of a company facing an irreversible investment decision in the presence of two sources of uncertainty. By paying a fixed sunk cost the company generates a continuous stochastic cash-flow, which results from selling two products on the market. In this framework, the company aims at maximizing its total expected profit arising from this investment and seeks to find a decision rule, which determines the optimal time to undertake this expenditure. \\
We will see that this amounts in solving a two-dimensional optimal stopping problem of the form
\begin{equation}\label{Introduction}
V(x,y) = \sup_{\tau } \mathbb{E} [ e^{-r \tau} F(X_{\tau}^{x} , Y_{\tau}^{y} ) ],
\end{equation}
where the supremum is taken over the set of stopping times and the function $F$ represents the value of the investment, dependent on two It$\hat{\text{o}}$-diffusions $X$ and $Y$ modelling the prices of the two products on the market (cf. (\ref{Objective J}) below). \\
\hspace*{0.2cm}Dating back to the seminal works of Myers  \cite{Myers} and McDonald and Siegel \cite{McDonald}, the real options approach to irreversible investment decisions has received much attention in economics and finance with various settings regarding the dimensionality and characteristics of the underlying stochastic process (cf.\,Dixit \cite{Dixit1}, Pindyck \cite{Pindyck1,Pindyck2} or Alvarez \cite{Alvarez}, Battauz et al.\,\cite{Battauz}, Luo et al.\,\cite{Luo} for more recent contributions). In the simplest form, where the underlying economic shock process is one-dimensional and the investment option gives rise to a perpetual payoff stream, explicit solutions are often feasible (cf.\,Dixit and Pindyck \cite{Dixit}, Stokey \cite{Stokey}, Trigeorgis \cite{Trigeorgis} for a survey). On the other hand, there are still only few examples of solvable multidimensional optimal stopping problems, despite the fact that real options models naturally deal with multiple sources of uncertainty. \\
\hspace*{0.2cm}In some models, the dimensionality of the problem can be effectively reduced to one. For instance, McDonald and Siegel \cite{McDonald} derive the optimal solution for the ratio of investment value and investment cost, and thus trace the problem back to a one-dimensional problem, for which an explicit solution could be found. This operative method was used and improved by a number of authors such as Gerber and Shiu \cite{Gerber}, Shepp and Shiryaev \cite{Shepp} as well as Thijssen \cite{Thijssen} (see also Christensen et al.\,\cite{Christensen} and references therein). Nevertheless, in presence of a constant sunk cost of investment, a reduction of dimension \`{a} la McDonald and Siegel \cite{McDonald} is typically not feasible, as the problem's value function fails to be homogeneous of degree one. \\ 
\hspace*{0.2cm}Characterizing the solution in optimal stopping/real options models where the state space cannot be reduced is a challenging task. 
Hu and \O ksendal \cite{Hu} as well as Olsen and Stensland \cite{Olsen} consider an investment problem involving a multidimensional geometric Brownian motion, but their given conjecture regarding the shape of the stopping region only holds true in trivial cases, as pointed out by Christensen and Irle \cite{Christensen2}. Adkins and Paxson \cite{Adkins} proposed a quasi-analytical approach, which results in solving a set of simultaneous equations, but their methodology seems to trigger sub-optimal solutions (see also Compernolle et al.\,\cite{Dutch}, Lange et al.\,\cite{Lange}).
There are, however, some recent contributions in which a complete characterization of the solution to truly multidimensional irreversible investment problems was derived. De Angelis et al.\,\cite{GF} study a singular stochastic control problem and the associated two-dimensional optimal stopping problem, for which they characterize the optimal boundary as the unique solution to a nonlinear integral equation. Christensen and Salminen \cite{ChrisSalm} propose a solution method relying on the Riesz representation of excessive functions and study a classical investment problem, for which they derive an integral equation of similar structure. In both references, the uniqueness of the representation is established by relying on arguments first presented in Peskir \cite{PeskirU}. \\
\hspace*{0.2cm}In this paper, we consider and solve optimal investment problem (\ref{Introduction}), which was first introduced by Compernolle et al.\,\cite{Dutch}. In that work, the authors derived some important preliminary results regarding the value function as well as the corresponding optimal boundary, but did not achieve a complete characterization of the latter. In this work, we push the analysis of Compernolle et al.\,\cite{Dutch} much further. Borrowing arguments from De Angelis et al.\,\cite{GF}, we determine an integral equation for the optimal investment boundary (cf. Theorem \ref{Theorem Integral Equation}). Moreover, we provide an analytical rigorous study of the dependency of the optimal boundary on some model's parameters. To our knowledge, such a result appears here for the first time. As a matter of fact, the analytical approach to comparative statics in Olsen and Stensland \cite{Olsen} (also employed by Compernolle et al.\,\cite{Dutch}) seems to overlook the delicate issue of the regularity of the value function. We are able to fix this issue by providing the proper regularity property, that in turn allows for a rigorous proof of the claimed monotonicity results and for additional findings (cf. Section \ref{Section Comparative Statics Analysis}).
Finally, inspired by the numerical analysis in Detemple and Kitapbayev \cite{Detemple} and Christensen and Salminen \cite{ChrisSalm}, we propose a probabilistic numerical approach for the determination of the optimal boundary through the derived integral equation. We provide details about the algorithm, with the aim of making a service to other studies dealing with related questions. It is worth noticing that the proposed probabilistic numerical method employs a Monte Carlo simulation, and  as such it does not face the curse of dimensionality, which is typical of analytical methods in large dimensions. \\
\hspace*{0.2cm}Overall, we believe that our main contributions are the following. From a mathematical point of view, given the limited amount of solvable multidimensional optimal stopping problems, we believe that our detailed study nicely complements the existing literature on optimal stopping as well as real options theory. Moreover, we suggest that our approach also has a methodological value for other real options problems. In fact, it defines an operative recipe for the determination of the optimal investment trigger analytically (by an integral equation) and numerically (by an approximation scheme), which can be easily adapted to different settings as well. \\
\hspace*{0.2cm}The paper is organized as follows. In Section 2 we introduce the optimal investment problem. In Section 3 we consider two benchmark problems, before we continue by characterizing the value function and the related optimal boundary in Sections 4 and 5. Analytical and numerical comparative statics results are then obtained in Section 6. Finally, some technical proofs and results are collected in the Appendices.
\section{The Irreversible Investment Problem}
\label{sec:setting}
Let $(\Omega, \mathcal{F}, \mathbb{F}:= (\mathcal{F}_{t})_{t \geq 0}, \mathbb{P})$ be a complete filtered probability space, with the filtration $\mathbb{F}$ generated by a two-dimensional Brownian motion $W = (W_{t}^{X}, W_{t}^{Y})_{t \geq 0 }$ and augmented with $\mathbb{P}$-null sets. We consider a profit-maximizing and risk-neutral company, which has the opportunity to invest into a production plant by paying a constant investment cost $I$. The production plant is capable of producing two goods in given quantities $Q_1$ and $Q_2$ and we assume that the prices of the two goods evolve stochastically according to the dynamics
\begin{equation} \label{dyn}
\begin{cases}
dX_{t}^{x} = \alpha_{1} X_{t}^{x} dt + \sigma_{1} X_{t}^{x} d W_{t}^{X}, & \quad X_0^x = x > 0, \\
dY_{t}^y \, = \alpha_{2} Y_{t}^y dt + \sigma_{2} Y_{t}^y dW_{t}^{Y}, & \quad Y_0^y = y\,  > 0,
\end{cases}
\end{equation}
for some constants $\alpha_{1}, \alpha_{2} \in \mathbb{R}$ and $\sigma_{1} , \sigma_{2} > 0$.
We assume that after the company has made the investment, it is able to sell the goods in their given quantities instantaneously and over an infinite time horizon on the market. If the investment is performed at initial time, its value for given price levels $x$ and $y$ is then obtained through the discounted perpetual revenue flow, net of the investment cost; that is,
\begin{equation}\label{defF(x,y)}
 \mathbb{E} \Bigl[ \int_{0}^{\infty} e^{-rt} \pi (X_{t}^{x} , Y_{t}^{y}) dt  -  I \Bigl] =: F(x,y) .
\end{equation}
Here $\pi(x,y) := Q_1 x + Q_2 y$ denotes the profit function and $r > 0$  is a discount factor. In order to guarantee finite integrals, we make the following \textbf{standing assumption}. 
\begin{assumption}\label{Assr>a}
We have $r > \alpha_{1} \lor \alpha_{2}$.
\end{assumption}
Clearly, an investment at initial time is not necessarily optimal. Hence, setting
\begin{equation}\label{Stopping times T}
\mathcal{T} := \{ \tau: ~  \tau \text{ are } \mathbb{F} \text{-stopping times}\},
\end{equation}
the company aims at determining the entry rule $\tau^* \in \mathcal{T}$ that maximizes its net total expected profits from $\tau^*$ on. That is, for any initial price levels $x,y \in \mathbb{R}_+ := [0, \infty)$, it seeks to determine $\tau^* \in \mathcal{T}$ such that 
\begin{align}\label{Value Function}
V(x,y) := \mathcal{J}(x,y,\tau^*) = \max_{\tau \in \mathcal{T}} \mathcal{J}(x,y,\tau),
\end{align}
where
\begin{align}\label{Objective J}
\mathcal{J}(x,y,\tau) := \mathbb{E} \Bigl[ e^{-r \tau} F(X_{\tau}^{x} , Y_{\tau}^{y} ) \Bigl] = \mathbb{E} \Bigl[ e^{-r \tau} \Bigl( \frac{Q_{1} X_{\tau}^{x}}{\delta_{1}} + \frac{Q_{2} Y_{\tau}^{y}}{\delta_{2}} - I \Bigl) \Bigl]
\end{align}
for $\delta_i = r - \alpha_i$, $i = 1,2.$ The last equality in (\ref{Objective J}) follows by straightforward  calculations upon using Assumption \ref{Assr>a}. Throughout this paper, we will refer to (\ref{Value Function}) as to \emph{the optimal investment problem}.
\begin{remark}\label{Remark Technicalities}
Assumption \ref{Assr>a} guarantees $
\mathbb{E} [  \sup_{t \geq 0} e^{-rt} X_{t}^{x} ] < + \infty$ and $ \mathbb{E} [ \sup_{t \geq 0} e^{-rt} Y_{t}^{y} ] < + \infty$, standard technical assumptions in the theory of optimal stopping (cf. Karatzas and Shreve \cite{Kara}, p.\,35). Amongst other things, these conditions imply that the families of random variables 
\begin{equation}\label{e -rt Xt uniformly integrable}
\{ e^{-r\tau}X_{\tau}^{x} \one_{\{\tau < \infty \} }  , \, \tau \in \mathcal{T} \}
\quad \text{and} \quad \{e^{-r\tau} Y_{\tau}^{y} \one_{\{\tau < \infty \} }, \, \tau \in \mathcal{T} \}
\end{equation} 
are uniformly integrable. Moreover $\lim_{t \to \infty} e^{-rt} X_{t}^{x} = 0$ as well as $\lim_{t \to \infty} e^{-rt} Y_{t}^{y} = 0$ $\mathbb{P}$-a.s., we thus adopt the convention
\begin{align}\label{convention lim e X eins=infty}
e^{-r \tau} X_{\tau}^{x} \one_{ \{ \tau = \infty \}} &:= \lim_{ t \to \infty} e^{-rt} X_{t}^{x} = 0, \quad \mathbb{P}\text{-a.s.} \nonumber \\  
\text{as well as} \quad e^{-r \tau}\, Y_{\tau}^{y} \one_{ \{ \tau = \infty \} } &:= \lim_{ t \to \infty} e^{-rt}\text{\,}  Y_{t}^{y} = 0, \quad \mathbb{P}\text{-a.s.} 
\end{align}
and set 
\begin{equation}\label{e -r tau f(X,Y) on tau=infty}
e^{-r \tau} \mid f(X_{\tau}^{x} , Y_{\tau}^{y}) \mid \, := \limsup_{t \to \infty} e^{-rt} \mid f(X_{t}^{x} , Y_{t}^{y} )\mid \quad \mathbb{P}\text{-a.s.} 
\end{equation}
on $\{ \tau = \infty \}$ for any Borel-measurable function $f$. 
\end{remark}
\section{Two Benchmark Problems} \label{Section Benchmark problem}
Before we study the optimal entry problem introduced in the previous section, it is useful to focus on two related classical real options problems. Notice that the values $x=0$ and $y=0$ are absorbing boundaries for the processes $X_t^x$ and $Y_t^y$. In particular, when $X_0^0 = 0$ (resp. $Y_0^0 =0$), then $X_t^0 =0$ (resp. $Y_t^0 =0$) for all $t \geq 0$ $\mathbb{P}$-a.s. Therefore, we can naturally associate to (\ref{Value Function}) the two one-dimensional optimal stopping problems 
\begin{align}\label{OSP one dimensional}
v_1(x)  :=   \sup_{\tau \in \mathcal{T}} \mathbb{E} \Bigl[ e^{-r\tau} \Bigl( \frac{Q_{1} X_{\tau}^{x} } {\delta_{1}} - I \Bigl) \Bigl] \quad \text{and} \quad 
v_2 (y) :=  \sup_{\tau \in \mathcal{T} } \mathbb{E} \Bigl[ e^{-r\tau} \Bigl( \frac{Q_{2} Y_{\tau}^{y} } {\delta_{2}} - I \Bigl) \Bigl] 
\end{align}
Due to the one-dimensional structure of this problem, their solution is standard and can be obtained by a \textit{guess-and-verify approach} (cf.\,Dixit and Pindyck \cite{Dixit}).   \\
Let us consider $v_1$, as analogous considerations can be made for $v_2$. It is reasonable to assume that the company invests into the production plant only when the current price of the product is large enough. We thus expect that the optimal stopping time for problem (\ref{OSP one dimensional}) is of the form 
\begin{align*}
\tau^*_x := \inf \{ t \geq 0:~ X_t^x \geq x^* \},
\end{align*}
where $x^*$ denotes the critical price level, at which the company decides to invest. Accordingly, the candidate value function $w$ should satisfy $(\mathcal{L}_X -r)w(x) = 0$ for all $x < x^*$, where $\mathcal{L}_X$ denotes the second-order differential operator (acting on twice-continuously differentiable functions) given by 
\begin{align}
\mathcal{L}_X := \frac{1}{2} \sigma_1^2 x^2 \frac{\partial}{\partial x^2} + \alpha_1 x \frac{\partial}{\partial x}. 
\end{align}
 It is well known that the equation $(\mathcal{L}_X -r)w(x)=0$ admits two fundamental solutions $\psi (x) = x^{\beta_1}$ and $\varphi (x) = x^{\beta_2}$, where $\beta_1$ and $\beta_2$ are the positive and negative solutions to the equation
\begin{align*}
\frac{1}{2} \sigma_1^2 \beta (\beta -1) + \alpha_1 \beta - r = 0
\end{align*}
and Assumption \ref{Assr>a} guarantees $\beta_1 > 1$. Consequently, any of its solutions takes the form $w(x) = A \psi (x) + B \varphi (x)$ for $x < x^*$, where $A$ and $B$ are constants to be found. As $x \mapsto \varphi (x)$ diverges as $x \downarrow 0$, and it is easy to see that $v_1$ has instead sublinear growth, we guess $B=0$. The candidate value function $w$ thus can be written as 
\begin{align*}
w(x) = \begin{cases}
Ax^{\beta_1} & x < x^* \\
\frac{Q_1 x}{\delta_1} - I & x \geq x^*
\end{cases}
\end{align*}   
for $A$ and $x^*$ to be derived. By employing the standard smooth-pasting and smooth-fit condition, it is straightforward to see that they are given by 
\begin{align}\label{A and xstar}
x^{*} = \frac{\beta_1}{(\beta_1 - 1) Q_{1}} \delta_{1} I \quad \text{and} \quad A = \frac{Q_{1}}{\beta_{1} \delta_{1}} x^{* 1-\beta_{1}}.
\end{align} 
The following proposition verifies that the candidate value function $w$ constructed in this way indeed coincides with the value function $v_1$ of (\ref{OSP one dimensional}). Its proof is standard and we refer to the classical textbook of Peskir and Shiryaev \cite{Peskir} for techniques and results. 
\begin{proposition}
Recall $v_1$ from (\ref{OSP one dimensional}). Then we have
\begin{align*}
v_1 (x) = 
\begin{cases}
A x^{\beta_1} &0 < x<x^{*}, \\
\frac{Q_{1} x}{\delta_{1}} - I &x \geq x^{*},
\end{cases}
\end{align*}
where $A$ and $x^*$ are given by (\ref{A and xstar}). Also,
\begin{align*}
\tau_{x}^{*}  :=  \inf \{ t \geq 0: ~ X_{t}^{x} \geq x^{*} \}
\end{align*}
is the optimal stopping time. 
\end{proposition}
Analogously, we have the next result concerning $v_2$.
\begin{proposition}
Recall $v_2$ from (\ref{OSP one dimensional}). Then
\begin{align*}
v_2 (y) = 
\begin{cases}
D y^{\eta_1} &0 < y<y^{*}, \\
\frac{Q_{2} y}{\delta_{2}} - I &y \geq y^{*},
\end{cases}
\end{align*}
where the constant $D$ and the investment threshold $y^*$ are given by 
\begin{align}\label{ystar}
y^* =  \frac{\eta_1}{(\eta_1 - 1) Q_{2}} \delta_{2} I \quad \text{and} \quad
D = \frac{Q_{2}}{\eta_{1} \delta_{2}} y^{* 1-\eta_{1}},
\end{align}
where $\eta_1 > 1$ denotes the positive root of the quadratic equation $\frac{1}{2} \sigma_2^2 \eta (\eta -1) + \alpha_2 \eta -r = 0$. Moreover, the optimal stopping time is of the form 
\begin{align*}
\tau_{y}^{*}  :=  \inf \{ t \geq 0: ~ Y_{t}^{y} \geq y^{*} \}
\end{align*}
\end{proposition} 
As expected, the optimal thresholds $x^*$ and $y^*$ will be shown in our subsequent analysis to identify the limits as $y \downarrow 0 $ and $x \downarrow 0$, respectively, of the curve triggering the optimal investment rule in \eqref{Value Function}. 

\section{On the Value Function of the Optimal Investment Problem}\label{Section Value Function}
Consistently with the two benchmark problems of last section, we can expect that also for Problem (\ref{Value Function}) it will be optimal to invest when the price processes $X$ and $Y$ are sufficiently large. However, differently to $v_1$ and $v_2$ as in (\ref{OSP one dimensional}), (\ref{Value Function}) defines a two-dimensional optimal stopping problem for which a \textit{guess-and-verify approach} is not feasible. Hence, in the following we will perform a direct study of $V$. After deriving some preliminary results, we move on by defining the associated continuation and stopping regions. The main result is then stated in Theorem \ref{Theorem Probabilistic Repres of v}, where we borrow arguments from De Angelis et al.\,\cite{GF} in order to derive a probabilistic representation of $V$. The proof of the next proposition can be found in Appendix \ref{Appendix Proof of properties of value function}.
\begin{proposition}\label{Proposition Properties of value function}
Recall $V$ from (\ref{Value Function}). There exists a constant $C > 0$ such that for all $(x,y) \in \mathbb{R}_+^2$
\begin{align}\label{F < V < C}
\max \{0, F(x,y) \}  \leq V(x,y) \leq C ( x + y),
\end{align}
and the value function $V$ is nondecreasing with respect to $x$ and $y$. Moreover, $V$ is continuous and convex on $\mathbb{R}_+^2$. 
\end{proposition}
\subsection*{Continuation and Stopping Regions.}\label{Subsection Continuation and Stopping Region}
As it is customary in optimal stopping, continuation and stopping regions of the optimal stopping problem (\ref{Value Function}) are given by
\begin{align}\label{Continuation and Stopping region}
\mathcal{C} := \{ (x,y) \in \mathbb{R}_{+}^{2}: ~ V(x,y) > F(x,y) \} \quad \quad
\mathcal{S} := \{(x,y) \in \mathbb{R}_{+}^{2}: ~ V(x,y) = F(x,y) \}.
\end{align}  
Notice that, since the value function $V$ and the function $F$ are continuous, the continuation region is open and the stopping region is closed (cf.\,Peskir and Shiryaev \cite{Peskir}, p.\,36).
Moreover, the optimal stopping time is given by the first entry time of the process $(X_{t}^{x} , Y_{t}^{y})$ into the stopping region
\begin{align} \label{Optimal stopping time}
\tau^* = \tau^*(x,y) := \inf \{ t \geq 0 : ~ (X_{t}^{x}, Y_{t}^{y}) \in \mathcal{S} \},
\end{align}
whenever it is $\mathbb{P}$-a.s. finite (cf.\,Peskir and Shiryaev \cite{Peskir}, p.\,46).
\subsection*{Probabilistic Representation of the Value Function.}\label{Subsection Probabilistic representation of the value function}
We now provide a probabilistic representation of the value function $V$ of the stopping problem (\ref{Value Function}). This representation is essential for the  forthcoming characterization of the optimal boundary being the solution to an integral equation. Its technical proof employs an approximation argument as in De Angelis et al.\,\cite{GF} and it is postponed to Appendix \ref{Appendix Proof of probabilistic repr.}.   
\begin{theorem}\label{Theorem Probabilistic Repres of v}
The value function $V$ of the optimal investment problem (\ref{Value Function}) admits the following representation.
\begin{align}\label{Probabilistic representation of v}
V(x,y) = \mathbb{E} \Bigl[ \int_{0}^{\infty} e^{-rt} (Q_{1} X_{t}^{x} + Q_{2} Y_{t}^{y} - rI) \one_{ \{ (X_{t}^{x} , Y_{t}^{y} ) \in \mathcal{S} \}} dt \Bigl].
\end{align}
for all $(x,y) \in \mathbb{R}_{+}^{2}$. 
\end{theorem} 
It is worth anticipating already here that representation \eqref{Probabilistic representation of v} will be employed in Proposition \ref{Proposition Value function C1} in order to prove that actually $V\in C^1(\mathbb{R}^2_+)$, and in Theorem \ref{Theorem Integral Equation} in order to determine a nonlinear equation that uniquely characterizes the free boundary triggering the optimal investment time.

\begin{remark}\label{Remark Martingale}
Let $H(x,y) := (Q_{1} x + Q_{2} y - rI) \one_{\{ (x,y) \in \mathcal{S} \}}$. The expression (\ref{Probabilistic representation of v}) can thus be formulated as 
\begin{align}\label{value function prob with H in Theorem}
V(x,y) = \mathbb{E} \Bigl[ \int_{0}^{\infty} e^{-rt} H(X_{t}^{x} , Y_{t}^{y}) dt \Bigl].
\end{align}
Notice that $\vert H(x,y) \vert \leq Q_1 x + Q_2 y - rI$. Upon using Assumption \ref{Assr>a}, the strong Markov property and standard arguments on conditional expectation, we have
\begin{align}\label{V and H identity}
\mathbb{E}\Big[ \int_0^\infty e^{-rt} H(X_t^x , Y_t^y ) dt \Big\vert \mathcal{F}_\tau \Big] = \int_0^\tau e^{-rs} H(X_s^x , Y_s^y) ds + e^{-r \tau} V(X_\tau^x , Y_\tau^y ).
\end{align}
Consequently, the process   
\begin{align}\label{V and H martingale}
\big\{ e^{-rt} V(X_t^x , Y_t^y ) + \int_0^t e^{-rs} H (X_s^x , Y_s^y)ds , \, t \geq 0 \big\}
\end{align}
is an $(\mathcal{F}_t)$-martingale. Furthermore, equation (\ref{V and H identity}) implies
\begin{align}
\vert e^{-rt} V(X_t^x , Y_t^y ) \vert  \leq  \mathbb{E} \Big[ \int_0^\infty e^{-rt} H(X_t^x , Y_t^y ) dt \Big\vert \mathcal{F}_\tau \Big]
\end{align}
and it follows that the family $\{ e^{-rt} V(X_\tau^x , Y_\tau^y ), \tau \in \mathcal{T} \}$ is uniformly integrable. 
\end{remark} 

\begin{remark}
\label{rem:XY}
Notice that the results of Theorem \ref{Theorem Probabilistic Repres of v} can be generalized to the case in which $X$ and $Y$ are general one-dimensional It$\hat{\text{o}}$-diffusions. Indeed, the arguments of the proof in Appendix \ref{Appendix Proof of probabilistic repr.} do not actually hinge on the particular form of the price processes. 
\end{remark}

\section{On The Optimal Boundary} 
In this section, we study the optimal price level triggering the investment in Problem (\ref{Value Function}). Some of the subsequent results have already been derived by Compernolle et al.\,\cite{Dutch}, Theorem 1, and we repeat them briefly for the sake of completeness. The main novel result is then stated in Theorem \ref{Theorem Integral Equation}, where we characterize the optimal trigger as the unique solution to a nonlinear integral equation in a certain functional class.   \\
Define
\begin{align}\label{Definition b(x)}
b(x) := \sup \{ y \in \mathbb{R}_{+}: ~ V(x,y) > F(x,y) \}, \quad x \in \mathbb{R}_{+},
\end{align}
with the convention $\sup \emptyset = 0 $. We state the following proposition. 
\begin{proposition}\label{Proposition rewrite C and S}
The continuation region and stopping region of (\ref{Continuation and Stopping region}) can be written as
\begin{align}\label{Rewrite Continuation and Stopping region}
\mathcal{C} = \{ (x,y) \in \mathbb{R}_{+}^{2}: ~ y < b(x) \}, \quad \quad \mathcal{S}= \{ (x,y) \in \mathbb{R}_{+}^{2}: ~ y \geq b(x) \}
\end{align}
\end{proposition}
\begin{proof}
It is sufficient to prove that the continuation region is down-connected. Take $(x,y) \in \mathcal{C}$ and $\tau^* (x,y)$ of (\ref{Optimal stopping time}). We then have
\begin{align*}
V(x,y) = \sup_{\tau \in \mathcal{T}} \mathbb{E} [ e^{-r \tau} F(X_{\tau}^{x} , Y_{\tau}^{y}) ] = \mathbb{E} [ e^{-r \tau^*} F(X_{\tau^*}^{x} , Y_{\tau^*}^{y})] > F(x,y)
\end{align*}
Let $\epsilon \in (0,y]$ and notice that $\tau^*(x,y)$ is a-priori suboptimal for the stopping problem with value function $V(x, y -\epsilon)$. It follows that
\begin{align*}
V(x, y- \epsilon) 
&\geq \mathbb{E} [ e^{-r \tau^*} F(X_{\tau^*}^{x} , Y_{\tau^*}^{y - \epsilon})] 
= \mathbb{E} [ e^{-r \tau^*} F(X_{\tau^*}^{x} , Y_{\tau^*}^{y})]
 - \frac{Q_{2} \epsilon}{\delta_{2}}\mathbb{E} [ e^{-r \tau^*} Y_{\tau^*}^{1} ] 
\\
&\geq V(x,y) - \frac{Q_{2} \epsilon}{\delta_{2}} 
> F(x,y) - \frac{Q_{2} \epsilon}{\delta_{2}} 
 = F(x, y - \epsilon)
\end{align*}
where we used that $\{ e^{-rt} Y_{t}^{y}, ~ t \geq 0 \}$ is a supermartingale due to (\ref{e -rt Xt uniformly integrable}) and Assumption \ref{Assr>a}. Hence, $(x, y- \epsilon) \in \mathcal{C}$ for every $\epsilon \in (0,y]$, which concludes our proof. 
\end{proof}
The next proposition states some preliminary results of the 
boundary (\ref{Definition b(x)}).
\begin{proposition}\label{Proposition b nonincreasing and rightcontinuous}
The function $b$ of (\ref{Definition b(x)}) inherits the following properties. 
\begin{enumerate}
\item[(i)] $x \mapsto b(x)$ is nonincreasing on $\mathbb{R}_{+}$, 
\item[(ii)] $x \mapsto b(x)$ is right continuous on $\mathbb{R}_{+}$. 
\end{enumerate} 
\end{proposition}
\begin{proof}
\textit{(i)} The proof follows in the same spirit as the proof of Proposition \ref{Proposition rewrite C and S}, with the roles of $x$ and $y$ reversed. \\
\hspace*{0.3cm} \textit{(ii)} The functions $V$ and $F$ are continuous on $\mathbb{R}_{+}^{2}$, consequently $b$ is lower-semicontinuous. Since it is nonincreasing by point \textit{(i)},  the claim follows. 
\end{proof}
The results stated in Propositions \ref{Proposition rewrite C and S} and \ref{Proposition b nonincreasing and rightcontinuous} guarantee that the continuation region $\mathcal{C}$ and the stopping region $\mathcal{S}$ are connected. Moreover, we can rewrite the optimal stopping time (\ref{Optimal stopping time}) due to (\ref{Rewrite Continuation and Stopping region}) and obtain
\begin{align}\label{Optimal stopping time with boundary}
\tau^{*} = \tau^{*} (x,y) := \inf \{ t \geq 0: ~ Y_{t}^{y} \geq b(X_{t}^{x}) \}
\end{align}
for any $(x,y) \in \mathbb{R}_{+}^{2}$.
Furthermore, the probabilistic representation (\ref{Probabilistic representation of v}) rewrites as
\begin{align}\label{Value function Probabilistic new}
V(x,y) = \mathbb{E} \Bigl[ \int_{0}^{\infty} e^{-rt} \Bigl( Q_{1} X_{t}^{x} + Q_{2} Y_{t}^{y} - rI \Bigl) \one_{\{  Y_{t}^{y} \geq b(X_{t}^{x})) \}} dt \Bigl].
\end{align} 
for any $(x,y) \in \mathbb{R}_{+}^2$.
In the next step, we prove $V \in C^{1}(\mathbb{R}_{+}^{2})$. As a by-product, we obtain the well known \textit{smooth-fit} condition across the free boundary, which states the continuity of $V_{x}$ as well as $V_{y}$ at $\partial \mathcal{C}$. To this end, it is important to bear in mind the following well known fact.
\begin{lemma}\label{Lemma Densities}
The processes $X^{x}$ and $Y^{y}$ are given by two independent geometric Brownian motions, hence they have a log-normal distribution with transition densities  
\begin{align*}
\rho_{1} (t,x, \psi) &= \frac{1}{\sigma_{1} \psi \sqrt{2 \pi t}} \exp \Bigl( - \frac{(\log \psi - \log x - (\alpha_{1} - \frac{1}{2} \sigma_{1}^{2}) t)^2}{2 \sigma_{1}^{2} t} \Bigl)\\
\rho_{2} (t,y, \eta) &= \frac{1}{\sigma_{2} \eta \sqrt{2 \pi t}} \, \exp \Bigl( - \frac{(\log \eta - \log y - ( \alpha_{2} - \frac{1}{2} \sigma_{2}^{2} )t)^2 }{2 \sigma_{2}^{2} t} \Bigl)
\end{align*}
for every $(t,x,\psi), (t, y, \eta) \in (0, \infty) \times \mathbb{R}_+ \times \mathbb{R}_+$. Moreover, 
\begin{enumerate}
\item[i)] $(t, \zeta, \xi) \mapsto \rho_{i} (t, \zeta, \xi)$ is continuous on $(0, \infty) \times \mathbb{R}_+ \times \mathbb{R}_+$ for i=1,2;
\item[ii)] Let $ \mathcal{K} \subset \mathbb{R}_{+}^2$ be a compact set. Then there exists some $q >1$, which is possibly depending on $\mathcal{K}$, such that
\begin{align*}
\int_{0}^{\infty} e^{-rt} \Bigl( \int_{\mathcal{K}} \mid \rho_{1} (t,x, \psi) \rho_{2} (t,y, \eta) \mid^q d\psi d \eta \Bigl)^{1/q} dt  < + \infty
\end{align*}
for all $(x,y) \in \mathcal{K}$.
\end{enumerate}
\end{lemma}

\begin{proposition}\label{Proposition Value function C1}
The value function of (\ref{Value Function}) is such that $V \in C^{1}(\mathbb{R}_{+}^{2})$. 
\end{proposition}
\begin{proof}
We can rewrite (\ref{Probabilistic representation of v}) with Lemma \ref{Lemma Densities} and obtain 
\begin{align*}
V(x,y) &= \int_{0}^{\infty} e^{-rt}  \int_{0}^{\infty } \rho_{1} (t,x, \psi) \int_{b(\psi)}^{\infty} (Q_{1} \psi + Q_{2} \eta - rI) \rho_{2} (t,y, \eta)  d \eta \, d \psi  \, dt.
\end{align*}
Due to (\ref{e -rt Xt uniformly integrable}) and Lemma \ref{Lemma Densities}, we are able to take derivates with respect to $x$ and $y$. Standard dominated convergence arguments then show that $\partial_{x} V(x,y)$ as well as $\partial_{y} V(x,y)$ are continuous for all $(x,y) \in \mathbb{R}_{+}^2$. 
\end{proof}

\begin{remark}
\label{rem:approach}
As it has become clear from the analysis developed so far, our approach employs an approximation procedure and PDE results in order to obtain representation \eqref{Probabilistic representation of v}, which then allows to show that actually $V \in C^{1}(\mathbb{R}_{+}^{2})$ (cf.\,Proposition \ref{Proposition Value function C1}). Such an approach is quite flexible since the general results from PDE theory that we employ do actually hold for more general dynamics and payoff functions (see also Remark \ref{rem:XY}).\\
However, it is worth noticing that also other solution methods can be followed as well. For example, one could rely on the more probabilistic approach \`a la Peskir and Shiryaev \cite{Peskir} and the recent results on the smooth-fit principle by De Angelis and Peskir \cite{DeAngelisPeskir} in order to show that: (i) the value function $V$ is convex and  continuous, stopping and continuation regions are not empty, and the optimal stopping time is given by the first entry time of $(X,Y)$ into the stopping region; (ii) stopping and continuation regions are separated by a right-continuous decreasing boundary $b$, and $V$ solves its corresponding free boundary problem, in particular being a classical solution to a linear elliptic PDE in the interior of the continuation region; (iii) the points of the boundary $\partial \mathcal{C}$ are probabilistically regular (in the sense that the process $(X,Y)$ immediately enters the stopping region, as soon as it hits its boundary), which in turn guarantees the validity of the smooth-fit; i.e.\ $V \in C^{1}(\mathbb{R}_{+}^{2})$; (iv) the free boundary $b$ is continuous, e.g.\ by Peskir \cite{PeskirC}; (v) finally, by using the convexity of $V$ and arguing as in Proposition \ref{Proposition Variational Inequality}, we could prove that $V$ has in fact the needed regularity to apply (a weak version of) Dynkin's formula and derive \eqref{Probabilistic representation of v}. The latter in turn gives the integral equation for $b$.
\end{remark}


\subsection*{Continuity of the Optimal Boundary.}
\label{Subsection Continuity of the Free Boundary}

In order to derive the continuity of the optimal boundary over the whole state-space, it would be sufficient to prove the left-continuity, as we already established the right-continuity of $b$ in Proposition \ref{Proposition b nonincreasing and rightcontinuous}. Nevertheless, we follow the arguments of Compernolle et al.\,\cite{Dutch} relying on the convexity of the optimal boundary. 

\begin{proposition}\label{Proposition S convex}
The stopping region $\mathcal{S}$ of (\ref{Continuation and Stopping region}),
 (equivalently of (\ref{Rewrite Continuation and Stopping region})) is convex on $\mathbb{R}_{+}^{2}$. 
\end{proposition}
\begin{proof}
Assume there exist $(x_{1} , y_{1}), (x_{2}, y_{2}) \in \mathcal{S}$ and  $\lambda \in (0,1)$ such that $
(x,y)  :=  \lambda (x_{1} , y_{1})  +  (1-\lambda)(x_{2} , y_{2})  \in \mathcal{C}$. Thus we must have $V(x,y) > F(x,y)$ as well as $V(x_{i}, y_{i}) = F(x_{i}, y_{i})$ for $i =1,2$. It follows that
\begin{align*}
V(x,y)  >  F(x,y)  
=  \lambda F(x_{1}, y_{1}) + (1- \lambda) F(x_{2} , y_{2})
=  \lambda V(x_{1}, y_{1}) + (1- \lambda) V(x_{2} , y_{2}),
\end{align*}
which contradicts the convexity of $V$, as seen in Proposition \ref{Proposition Properties of value function}.
\end{proof}

\begin{proposition}\label{Proposition b convex}
The optimal boundary $b$ of (\ref{Definition b(x)}) is convex on $\mathbb{R}_{+}$. 
\end{proposition}
\begin{proof}
Notice that the stopping region $\mathcal{S}  =  \{ (x,y) \in \mathbb{R}_{+}^{2}: ~ y  \geq  b(x) \}$
is the epigraph of $b$. Due to Proposition \ref{Proposition S convex} it follows from standard results (see for example Borwein and Lewis \cite{Borwein}, p.\,43) that the boundary is convex on $\mathbb{R}_{+}$. 
\end{proof}

\begin{proposition}\label{Proposition boundary continuous}
The optimal boundary $b$ of (\ref{Definition b(x)}) is continuous on $\mathbb{R}_{+}$. 
\end{proposition}
\begin{proof}
The continuity of $b$ on $(0,\infty)$ follows from Proposition \ref{Proposition b convex}, as $b$ is convex on an open set.  
It remains to show that the boundary is continuous in $x=0$. Assume that $b(0)   \neq  b(0+)$. While $b(0) < b(0+)$ is a contradiction to Proposition \ref{Proposition b nonincreasing and rightcontinuous} as $b$ is nonincreasing, supposing that $b(0) > b(0+)$ contradicts the closedness of the stopping region $\mathcal{S}$. The boundary $b$ is thus continuous on $\mathbb{R}_+$. 
\end{proof}

\begin{remark}
As already noticed in Remark \ref{rem:approach}, continuity of $b$ could be also obtained by relying on a result by Peskir \cite{PeskirC}, which connects the principle of smooth-fit with the continuity of the boundary for models with two-dimensional diffusions. 
\end{remark}

It now becomes clear to what extent the optimal investment problem (\ref{Value Function}) is related to the benchmark problems we studied in Section \ref{Section Benchmark problem}. 
Since the optimal boundary is continuous, convex and nonincreasing on $\mathbb{R}_{+}$, it follows that $b(0) = y^*$ and $b(x) < y^{*}$ for all $x > 0$. Furthermore, due to the fact that $x^{*}$ is the solution to the optimal stopping problem on $\mathbb{R}_{+} \times \{ 0 \}$, the boundary (\ref{Definition b(x)}) is such that $b(x) = 0$ for $x \geq x^{*}$. The solutions to the benchmark problems therefore give the investment thresholds for the company at the $x$- and $y$-axis.  
\subsection*{An Integral Equation for the Optimal Boundary.} \label{Subsection Deriving an integral equation}
In this section, we aim at characterizing the optimal boundary $b$ as the unique solution to an integral equation in a certain functional class. For that purpose, we make use of the probabilistic representation of the value function $V$ developed in Theorem \ref{Theorem Probabilistic Repres of v}. As a first step, we derive a lower bound for $b$ of (\ref{Definition b(x)}). Notice that Dynkin's formula implies 
\begin{align}\label{Dynkin on function F}
\mathbb{E} [ e^{-r \tau} F(X_{\tau}^{x} , Y_{\tau}^{y})] =  F(x,y) + \mathbb{E} \Bigl[ \int_{0}^{\tau} e^{-rs} ( \mathcal{L} -r)  F(X_{s}^{x} , Y_{s}^{y}) ds \Bigl]
\end{align}
for any bounded stopping time $\tau$. By localization arguments and (\ref{e -r tau f(X,Y) on tau=infty}), we conclude that (\ref{Dynkin on function F}) holds for any $\tau \in \mathcal{T}$ and we can thus rewrite $V$ as 
\begin{align*}
V(x,y) = F(x,y) + \sup_{\tau \in \mathcal{T}} \mathbb{E} \Bigl[ \int_{0}^{\tau} e^{-rs} (\mathcal{L} - r) F(X_{s}^{x} , Y_{s}^{y} ) ds \Bigl].
\end{align*}
Observe that it is never optimal to stop whenever $(\mathcal{L} - r) F(x,y) > 0$, consequently we have 
\begin{align}\label{h leq b by this sets}
\{ (x,y) \in \mathbb{R}_{+}^{2} : ~ (\mathcal{L} - r)F(x,y) > 0 \}  \subseteq  \{ (x,y) \in \mathbb{R}_{+}^{2}: ~ V(x,y) > F(x,y) \}  = \mathcal{C}.
\end{align}
Define
\begin{align}\label{Definition h(x)}
h(x) := \sup \{ y \in \mathbb{R}_{+} : ~ (\mathcal{L } - r) F(x,y) > 0 \}.
\end{align}
We state the following result.
\begin{lemma}\label{Lemma h(x)}
The function $h$ of (\ref{Definition h(x)}) is nonincreasing, continuous and it is given by the unique solution to the equation $(\mathcal{L} - r)F(x, \cdot) = 0$. In particular, 
\begin{align}\label{subseth(x)}
\{ (x,y) \in \mathbb{R}_{+}^{2}: ~ y > h(x) \}  =  \{ (x,y) \in \mathbb{R}_{+}^{2}: ~ rI - Q_{1}x - Q_{2}y < 0 \}.
\end{align}
\end{lemma}
\begin{proof}
We set $g(x,y) := (\mathcal{L} - r) F(x,y) = rI - Q_{1}x - Q_{2} y.$ Notice that $g$ is strictly decreasing and continuous in $x$ and $y$. Hence, for $x_{2} > x_{1}$ we have $g(x_{1} ,h(x_{2})) \geq g(x_{2} , h(x_{2})) \geq 0$, where the latter inequality is due to (\ref{Definition h(x)}), and it follows that $h(x_{1}) \geq h(x_{2})$. The continuity of $g$ and (\ref{Definition h(x)}) guarantee that $h$ solves $g(x, \cdot) =0$. Furthermore, as $g(x, \cdot)$ is strictly decreasing, $h$ the unique solution. Consequently, it admits the representation
\begin{align*}
h(x) = \frac{1}{Q_{2}} (rI - Q_{1} x).
\end{align*}
It is evident that $h$ is continuous on $\mathbb{R}_{+}$ and (\ref{subseth(x)}) follows from the above results. 
\end{proof}
Consider the class of functions
\begin{align*}
\mathcal{M} := \{ f: \mathbb{R} \mapsto \mathbb{R} , \text{ continuous, decreasing and s.t.} ~ f(x) \geq h(x) \}
\end{align*}
and notice that $\mathcal{M}$ is nonempty as $h \in \mathcal{M}$ due to Lemma \ref{Lemma h(x)}. 
\begin{theorem}\label{Theorem Integral Equation}
The optimal boundary $b$ of (\ref{Definition b(x)}) is the unique function  $y \in \mathcal{M}$ such that
\begin{align}\label{Integral Equation w/o densities}
F(x, y(x)) = \mathbb{E} \Bigl[ \int_{0}^{\infty} e^{-rt} \Bigl( Q_{1} X_{t}^{x} + Q_{2} Y_{t}^{y(x)} - rI \Bigl) \one_{\{  Y_{t}^{y(x)} \geq y(X_{t}^{x}) \}} dt \Bigl], \quad x>0.
\end{align} 
Equivalently, with regards to Lemma \ref{Lemma Densities}, one has
\begin{align}\label{Integral Equation w Densities}
\frac{Q_{1}x}{\delta_{1}} + \frac{Q_{2} y(x)}{\delta_{2}} - I 
= \int_{0}^{\infty} e^{-rt} \Big( \int_{0}^{\infty} \rho_{1} (t,x ,\psi) 
\Big( \int_{y(\psi)}^{\infty} (Q_{1} \psi + Q_{2} \eta - rI ) \rho_{2} (t , y(x) , \eta) d \eta \Big) d \psi \Big) dt.
\end{align} 
\end{theorem}
\begin{proof} As for the existence, it is sufficient to show that $b$ of (\ref{Definition b(x)}) solves the equation. Notice that $b \in \mathcal{M}$ due to Proposition \ref{Proposition b nonincreasing and rightcontinuous}, Proposition \ref{Proposition boundary continuous} and simple comparison arguments resulting from (\ref{h leq b by this sets}). Furthermore, by evaluating both sides of the probabilistic representation  (\ref{Probabilistic representation of v}) of $V$ at points $y = b(x)$, one finds (\ref{Integral Equation w/o densities}), upon using $V(x,b(x)) = F(x, b(x))$. \\
In order to show that $b$ is the unique solution to (\ref{Integral Equation w/o densities}) in $\mathcal{M}$, one can adopt the four-step procedure in De Angelis et al.\,\cite{GF}, extending and refining the original probabilistic arguments from Peskir \cite{PeskirU}. 
\end{proof}
\begin{remark}
The equation (\ref{Integral Equation w Densities}) can be reformulated in the canonical Fredholm form. Define
\begin{align*}
K(x, \psi, \alpha, \beta) = \int_{0}^{\infty} e^{-rt} \rho_{1} (t,x, \psi) \int_{\beta}^{\infty} \Bigl(Q_{1} \psi + Q_{2}\eta -rI\Bigl) \rho_{2} (t, \alpha, \eta) d \eta dt
\end{align*}
and after applying Fubini's theorem the equation (\ref{Integral Equation w Densities}) can be written as
\begin{align*}
\frac{Q_{1}x}{\delta_{1}} + \frac{Q_{2} b(x)}{\delta_{2}} - I = \int_{0}^{\infty} K(x, \psi, b(x) , b(\psi)) d\psi.
\end{align*}
We thus obtain the representation 
\begin{align}\label{Fredholm integral equation}
b(x) = f(x) + \underbrace{\lambda \int_{0}^{\infty} K(x, \psi, b(x) , b(\psi)) d\psi}_{:= G},
\end{align}
where we set 
\begin{align*}
\lambda := \frac{\delta_{2}}{Q_{2}} \hspace{1cm} \text{and} \hspace{1cm} f(x) := \frac{\delta_{2}}{Q_{2}} \Bigl( I - \frac{Q_{1} x}{\delta_{1}} \Bigl).
\end{align*}
Following Press and Teukolsky \cite{Press}, (\ref{Fredholm integral equation}) is a nonlinear, inhomogeneous Fredholm integral equation of second kind. \\
It is interesting to notice that $b(x) \geq f(x)$ for all $x \in \mathbb{R}_+$, where $f(x)$ represents the price of the second product that makes the company \textit{indifferent} between investing and passing up on the investment opportunity. However, as the company wants to maximize its expected profit, it  aims to invest at a larger price level of the second product. Consequently, it adds the quantity $G$, which is strictly positive due to (\ref{h leq b by this sets}).
\end{remark}
\section{Comparative Statics Analysis}\label{Section Comparative Statics Analysis}
In this section we perform some comparative statics analysis of the value function $V$ and the optimal boundary $b$ of (\ref{Definition b(x)}). 
Differently to the majority of the contributions on real options problems, we are able to propose rigorous analytical proofs of the dependency of the value function $V$ on $\sigma_i, \, i=1,2$ and $\alpha_i, \, i=1,2$. Moreover, we implement a recursive numerical method to investigate the sensitivity of the optimal boundary with respect to the model's parameters. \\
The next important technical proposition will be used in Propositions \ref{Proposition Comp Stat Sigma} and \ref{Proposition Comp Stat Alpha}. Its proof can be found in Appendix \ref{Appendix Lemma Variational Inequality}.  
\begin{proposition}\label{Proposition Variational Inequality}
The value function $V$ of (\ref{Value Function}) is such that $V \in C^{1} (\mathbb{R}^2_+)\cap \mathcal{W}^{2,2}_{\text{loc}} (\mathbb{R}_+^2)$ and satisfies the variational inequality 
\begin{align*}
\max \{ (\mathcal{L} -r) v(x,y) , \, F(x,y) - v(x,y) \} = 0
\end{align*}
for a.e. $(x,y) \in \mathbb{R}_{+}^{2}$. 
\end{proposition}
The next result exploits the convexity and regularity of $V$ in order to prove monotonicity of $V$ with respect to $\sigma_1$. The same rationale has already been employed in Olsen and Stensland \cite{Olsen}, where, however, the delicate issue of the regularity of $V$ seems to be overlooked (as a matter of fact, the optimal value in Olsen and Stensland \cite{Olsen} is implicitly assumed to be of class $C^2$). 
\begin{proposition}\label{Proposition Comp Stat Sigma}
Let $\hat{\sigma}_1 > \sigma_1$ and let $\hat{X}_t^x$ denote the unique solution to (\ref{dyn}) when volatility is $\hat{\sigma}_1$. Moreover, let $\hat{V}$ denote the value function of (\ref{Value Function}) with underlying state $(\hat{X}_t^x , Y_t^y)$. It follows that 
\begin{align*}
\hat{V} (x,y) \geq V(x,y),  
\end{align*}
for all $(x,y) \in \mathbb{R}_+^2$. 
\end{proposition}
\begin{proof}
Recall $\mathcal{L}$ as in (\ref{L infinitesimal generator}) and let $\hat{\mathcal{L}}$ denote the infinitesimal generator when volatility is $\hat{\sigma_1}$. 
Notice that Proposition \ref{Proposition Variational Inequality} implies $(\mathcal{L}-r)V(x,y) \leq 0$ as well as $(\hat{\mathcal{L}} - r) \hat{V} (x,y) \leq 0$ for a.e. $(x,y) \in \mathbb{R}_{+}^{2}$.
Moreover, the fact that $V \in \mathcal{W}^{2,2}(\mathbb{R}_+^2)$ by Proposition \ref{Proposition Variational Inequality} implies that we can argue as in the proof of Proposition \ref{Proposition vn in W2p Qn} in order to apply Dynkin's formula and obtain
\begin{align}\label{Dynkin on V hat}
\mathbb{E}[ e^{-r\tau^{*}} \hat{V}(X_{\tau^{*}}^{x}, Y_{\tau^{*}}^{y} ) ] = \hat{V}(x,y) + \mathbb{E} \Big[ \int_{0}^{\tau^{*}} e^{-rs} (\mathcal{L} -r) \hat{V}(X_{s}^{x}, Y_{s}^{y}) ds \Big],
\end{align}
where $\tau^* = \tau^{*} (x,y)$ denotes the optimal stopping time for the stopping problem with value function $V(x,y)$. 
From (\ref{Dynkin on V hat}) we obtain by simple manipulation that 
\begin{align*}
\mathbb{E}[ e^{-r\tau^{*}} \hat{V}(X_{\tau^{*}}^{x},& Y_{\tau^{*}}^{y} ) ] \\ &= \hat{V}(x,y) + 
\mathbb{E} \Big[ \int_{0}^{\tau^{*}} e^{-rs} \Big( (\hat{\mathcal{L}} -r) \hat{V}(X_{s}^{x}, Y_{s}^{y}) + (\mathcal{L} - \hat{\mathcal{L}} ) \hat{V} (X_{s}^{x}, Y_{s}^{y}) \Big) ds \Big] \\
&\leq \hat{V}(x,y) +  \mathbb{E} \Big[ \int_{0}^{\tau^{*}} e^{-rs}  (\mathcal{L} - \hat{\mathcal{L}} ) \hat{V} (X_{s}^{x}, Y_{s}^{y}) ds \Big] \\ 
&= \hat{V}(x,y) + \mathbb{E} \Big[ \int_{0}^{\tau^{*}} e^{-rs} \Big( \frac{1}{2} (\sigma_1^2 - \hat{\sigma}_1^2 ) x^2 \hat{V}_{xx} (X_{s}^{x}, Y_{s}^{y}) \Big) ds \Big] \\
&\leq \hat{V}(x,y)
\end{align*}
where the latter inequality follows from the convexity of $V$. By noticing that Proposition \ref{Proposition Variational Inequality} implies $\hat{V} \geq F$ on $\mathbb{R}_+^2$, we then obtain 
\begin{align*}
\hat{V}(x,y) \geq \mathbb{E} [ e^{-r \tau^*} \hat{V} (X_{\tau^*}^{x} , Y_{\tau^*}^{y} ) ] \geq \mathbb{E} [ e^{-r \tau^*} F(X_{\tau^*}^{x} , Y_{\tau^*}^{y} )] = V(x,y),
\end{align*}
which concludes our claim.
\end{proof}
\begin{proposition}\label{Proposition Comp Stat Alpha}
Let $\hat{\alpha}_1 > \alpha_1$ and let $\hat{X}_t^x$ denote the unique solution to (\ref{dyn}) when the drift is $\hat{\alpha}_1$. Furthermore, let $\hat{V}$ denote the value function of (\ref{Value Function}) when the underlying state is $(\hat{X}_t^x , Y_t^y )$. We have 
\begin{align*}
\hat{V} (x,y) \geq V(x,y)
\end{align*}
for all $(x,y) \in \mathbb{R}_+^2$. 
\end{proposition}
\begin{proof}
We can argue similarly as in Proposition \ref{Proposition Comp Stat Sigma}. Let $\hat{\mathcal{L}}$ be as in (\ref{L infinitesimal generator}) but with drift coefficient $\hat{\alpha}_1$. Proposition \ref{Proposition Variational Inequality} again implies 
$(\mathcal{L}-r)V(x,y) \leq 0$ as well as $(\hat{\mathcal{L}} - r) \hat{V} (x,y) \leq 0$ for a.e. $(x,y) \in \mathbb{R}_{+}^{2}$. Moreover, due to Proposition \ref{Proposition Variational Inequality} we can argue as in the proof of Proposition \ref{Proposition vn in W2p Qn} in order to apply Dynkin's formula and obtain 
\begin{align}\label{Dynkin 2 on V hat}
\mathbb{E}[ e^{-r\tau^{*}} \hat{V}(X_{\tau^{*}}^{x}, Y_{\tau^{*}}^{y} ) ] = \hat{V}(x,y) + \mathbb{E} \Big[ \int_{0}^{\tau^{*}} e^{-rs} (\mathcal{L} -r) \hat{V}(X_{s}^{x}, Y_{s}^{y}) ds \Big],
\end{align}
where $\tau^* = \tau^{*} (x,y)$ denotes the optimal stopping time for the stopping problem with value function $V(x,y)$. It follows that 
\begin{align*}
\mathbb{E}[ e^{-r\tau^{*}} \hat{V}(X_{\tau^{*}}^{x},& Y_{\tau^{*}}^{y} ) ] \\ &= \hat{V}(x,y) + 
\mathbb{E} \Big[ \int_{0}^{\tau^{*}} e^{-rs} \Big( (\hat{\mathcal{L}} -r) \hat{V}(X_{s}^{x}, Y_{s}^{y}) + (\mathcal{L} - \hat{\mathcal{L}} ) \hat{V} (X_{s}^{x}, Y_{s}^{y}) \Big) ds \Big] \\
&\leq \hat{V}(x,y) +  \mathbb{E} \Big[ \int_{0}^{\tau^{*}} e^{-rs}   (\mathcal{L} - \hat{\mathcal{L}} ) \hat{V} (X_{s}^{x}, Y_{s}^{y})   ds \Big] \\ 
&= \hat{V}(x,y) + \mathbb{E} \Big[ \int_{0}^{\tau^{*}} e^{-rs} \Big( (\alpha_1 - \hat{\alpha}_1) x \hat{V}_{x} (X_{s}^{x}, Y_{s}^{y}) \Big) ds \Big] \\
&\leq \hat{V}(x,y),
\end{align*}
for all $(x,y) \in \mathbb{R}_+^2$, upon using that $\hat{V}$ is nondecreasing by Proposition \ref{Proposition Properties of value function}. We now write $F(x,y ; \alpha_1)$ in order to emphasize the dependency of $F$ on the drift coefficient $\alpha_1$. Notice that $ F(\cdot \,; \hat{\alpha}_1) \geq F(\cdot \, ; \alpha_1)$. Repeating arguments as in the proof of Proposition \ref{Proposition Comp Stat Sigma}, we obtain
\begin{align*}
\hat{V}(x,y) \geq \mathbb{E} [ e^{-r \tau^*} \hat{V} (X_{\tau^*}^{x} , Y_{\tau^*}^{y} ) ] \geq \mathbb{E} [ e^{-r \tau^*} F(X_{\tau^*}^{x} , Y_{\tau^*}^{y}; \hat{\alpha}_1 )] \geq \mathbb{E} [ e^{-r \tau^*} F(X_{\tau^*}^{x} , Y_{\tau^*}^{y}; \alpha_1 )] = V(x,y),
\end{align*}
where $\tau^* = \tau^{*} (x,y)$ again denotes the optimal stopping time for the stopping problem with value function $V(x,y)$.
\end{proof}
The results of Propositions \ref{Proposition Comp Stat Sigma} and \ref{Proposition Comp Stat Alpha} remain valid for a change in the coefficients $\alpha_2$ and $\sigma_2$.
\begin{corollary}\label{Corollary Comp stat a2 o2}
The value function $V$ of (\ref{Value Function}) is increasing in $\alpha_2$ as well as $\sigma_2$.
\end{corollary}

\subsection*{Numerical evaluation}
In the following, we implement a numerical scheme in order to determine the optimal investment boundary $b$ and to investigate its sensitivity with respect to the parameters $r$, $\alpha_i$, and $\sigma_i$, for $i=1,2$. Our scheme relies on the integral equation \eqref{Integral Equation w/o densities} uniquely solved by $b$, and on an application of the Monte-Carlo method. As such, the method can be efficiently employed in problems with dimension larger than two as well, whenever an integral equation for the free boundary can be derived. Alternative numerical methods are clearly possible as well, and in fact employed in the related literature.
For instance, Compernolle et al.\ \cite{Dutch} propose a finite difference scheme for the numerical approximation of the variational inequality associated to the problem's value function. However, as it is typical for analytical methods, this approach suffers from the curse of dimensionality so that an efficient approximation in larger dimensions can become problematic. Lange et al.\ \cite{Lange} propose a scheme in which the decision maker is only permitted to exercise the option at a set of Poisson arrival times that arrive at a finite rate. The project's value is then defined as the ``fixed point'' of an iterative scheme, where each iteration adds a single Poisson arrival time at which the decision maker is able to stop. This procedure defines a monotonically increasing sequence of lower bounds of the project's value, which can thus be found -- if it is bounded -- as the limit of this sequence. While this approach seems suitable also for problems in large dimensions, the exogenous Poisson process, however, adds a restriction, which can be seen as a liquidity constraint.

Recall that $b$ uniquely solves \eqref{Integral Equation w/o densities}. Let $\zeta$ be an auxiliary exponentially distributed random variable with parameter $r$ that is independent of $(W^X,W^Y)$. It follows that (\ref{Integral Equation w/o densities}) can be reformulated as 
\begin{align}\label{Integral equation in Comp Stat}
b(x) &= f(x) + \lambda \mathbb{E} \Bigl[  \int_0^\infty e^{-rt} \Bigl( Q_{1} X_{t}^{x} + Q_{2} Y_{t}^{b(x)} - rI \Bigl) \one_{\{  Y_{t}^{b(x)} \geq b(X_{t}^{x}) \}} dt  \Bigl] \nonumber \\
&= f(x) + \lambda \frac{1}{r} \mathbb{E} \Bigl[  \Bigl( Q_{1} X_{\zeta}^{x} + Q_{2} Y_{\zeta}^{b(x)} - rI \Bigl) \one_{\{  Y_{\zeta}^{b(x)} \geq b(X_{\zeta}^{x}) \}}  \Bigl].
\end{align}
The latter representation is useful, as it allows for an application of Monte-Carlo methods in order to estimate expectations. In the following, we apply an iterative procedure, inspired by the contributions of Christensen and Salminen \cite{ChrisSalm} and Detemple and Kitapbayev \cite{Detemple}. For $(x,y) \in \mathbb{R}_+^2$ and a function $b: \mathbb{R}_+ \to \mathbb{R}_+$ , we define the operator 
\begin{align}\label{Operator Definition}
\Psi (x, y; b) :=  f(x) + \lambda \frac{1}{r} \mathbb{E} \Bigl[  \Bigl( Q_{1} X_{\zeta}^{x} + Q_{2} Y_{\zeta}^{y} - rI \Bigl) \one_{\{  Y_{\zeta}^{y} \geq b(X_{\zeta}^{x}) \}}  \Bigl].
\end{align}
It follows that the equation (\ref{Integral equation in Comp Stat}) rewrites as a fixed point problem
\begin{align}\label{Fixed point problem}
b(x) = \Psi (x, b(x) ; b), \quad x \in \mathbb{R}_+ , 
\end{align}
which we aim to solve by an iterative scheme. In order to do so, we define the sequence of boundaries 
\begin{align}\label{Iterative scheme}
b^{(n)} (x) = \Psi ( x, b^{(n-1)}(x) ; b^{(n-1)} ), \quad x \in \mathbb{R}_+ ,
\end{align}
for $n \geq 1$ and choose the initial boundary $b^{(0)}$ such that $b^{(0)} (0) = y^* , \, b^{(0)}  (x^*) = 0, \, b^{(0)} (x^*)$ is the vertex of a parabola and $b^{(0)} (x) = 0$ for all $x \geq x^*$. Moreover, for a given boundary $b^{(k)}$ we  estimate the expectation in (\ref{Operator Definition}) by 
\begin{align}\label{Monte Carlo}
\frac{1}{N} \sum_{i=1}^{N} \big( Q_1 X_{\zeta_i }^{i , x} + Q_2 Y_{\zeta_i }^{i, b^{(k)} (x) } - rI \big) \one_{ \big\{ Y_{\zeta_i }^{i, b^{(k)} (x) } \geq b^{(k)} (X_{\zeta_i }^{i , x}) \big\} } , 
\end{align}
where $N$ is the total amount of implemented realizations of an exponential random variable with parameter $r$. Consequently, for each $i = 1,...,N$, $\zeta_i$ denotes the value of time, while $X_{\zeta_i}^{i,x}$ and $Y_{\zeta_i}^{i,y}$ are the prices of the two products. Under the described procedure, the scheme (\ref{Iterative scheme}) is then iterated until the variation between steps falls below a predetermined level. 
\begin{remark}
In principle, the suggested numerical approach is suitable for a general class of optimal stopping problems for which an integral equation for the free boundary can be derived (see, e.g., Christensen and Salminen \cite{ChrisSalm}, Cai et al.\ \cite{Cai}). If an educated initial guess regarding the shape of the free boundary is possible, the algorithm seems to converge fast to an approximate solution of the integral equation. Notice also that the suggested method does not necessarily rely on the function $F(x,\,\cdot)$ to be linear or invertible. As a matter of fact, for a general payoff function $F$ one can (trivially) rewrite \eqref{Integral Equation w/o densities} as
\begin{align}
\label{eq:inteqgen}
F(x, b (x)) + b(x) &= b(x) + \mathbb{E}\Big[ \int_0^\infty e^{-rt} K (X_t^x , Y_t^{b(x)}) \one_{ \{ Y_t^{b(x)} \geq b(X_t^x) \} } dt \Big],
\end{align}
and then implement the iterative scheme 
\begin{align*}
b^{(n)}(x) = \Psi (x, b^{(n-1)}(x) , b^{(n-1)} ), \quad x \in \mathbb{R},
\end{align*}
where the operator now takes the slightly different form
\begin{align*}
\Psi (x,y;b) := y - F(x, y) + \frac{1}{r} \mathbb{E} \Big[ K(X_{\zeta}^x , Y_{\zeta}^{y}) \one_{ \{ Y_\zeta^{y} \geq b(X_\zeta^x) \} } \Big].
\end{align*}
Moreover, as we rely on an application of the Monte-Carlo method in order to evaluate the expected value on the right-hand side of \eqref{eq:inteqgen}, the proposed algorithm does not require the processes $X$ and $Y$ to have known densities $\rho_1$ and $\rho_2$, respectively. An Euler approximation of the dynamics of $X$ and $Y$ could indeed be used in order to simulate the random variable $(X_{\zeta}^x , Y_{\zeta}^{b(x)})$.
\end{remark}

\begin{figure}[H]
\begin{center}
  \includegraphics[width=9cm]{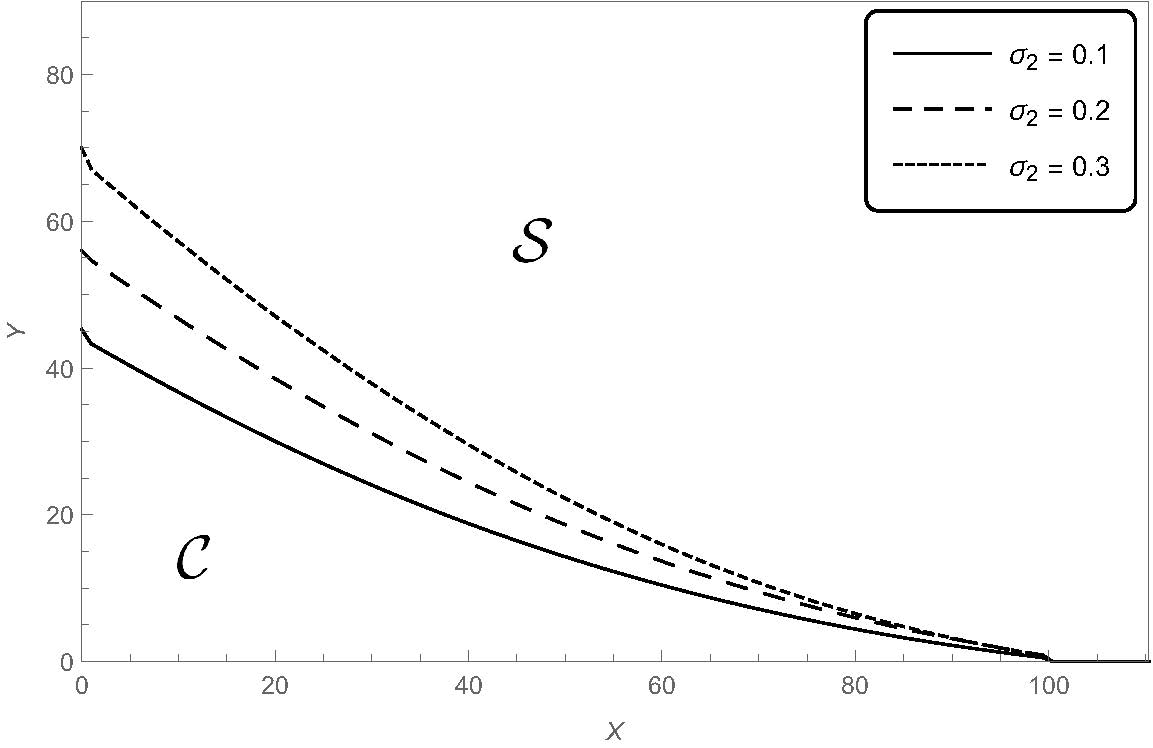}
  \caption{The optimal boundary for different values of $\sigma_2$ and following parameters: $r =0.1 ,\, \alpha_1 = 0.03 , \, \alpha_2 = 0.03 ,\, \sigma_1 = 0.15 ,\, Q_1 =5 ,\, Q_2 = 10 ,\, I = 4000$} \label{Figure: Comp Stat Sigma}
\end{center}
\end{figure} 
In Figure \ref{Figure: Comp Stat Sigma} we can observe the dependency of the optimal boundary $b$ with respect to $\sigma_2$. It is evident, that the boundary increases with $\sigma_2$. A larger  volatility coefficient may be interpreted as a higher level of uncertainty, which is equivalent to a higher price fluctuation in our model. The price thus has larger distortions in the downward direction -- but also upwards. The firm exploits the latter fact and thus waits for higher prices to evolve. Furthermore, this effect also results in a larger expected profit of the firm, which also follows by Proposition \ref{Proposition Comp Stat Sigma}. Notice that the threshold value $x^{*}$ does not change in Figure \ref{Figure: Comp Stat Sigma},  as it depends exclusively on the parameters $Q_{1}, r, \alpha_1$ and $\sigma_1$. A change in the volatility coefficient $\sigma_{2}$ thus has no influence on the investment threshold on the $x$-axis. \\ 
\begin{figure}[H]
\begin{center}
  \includegraphics[width=9cm]{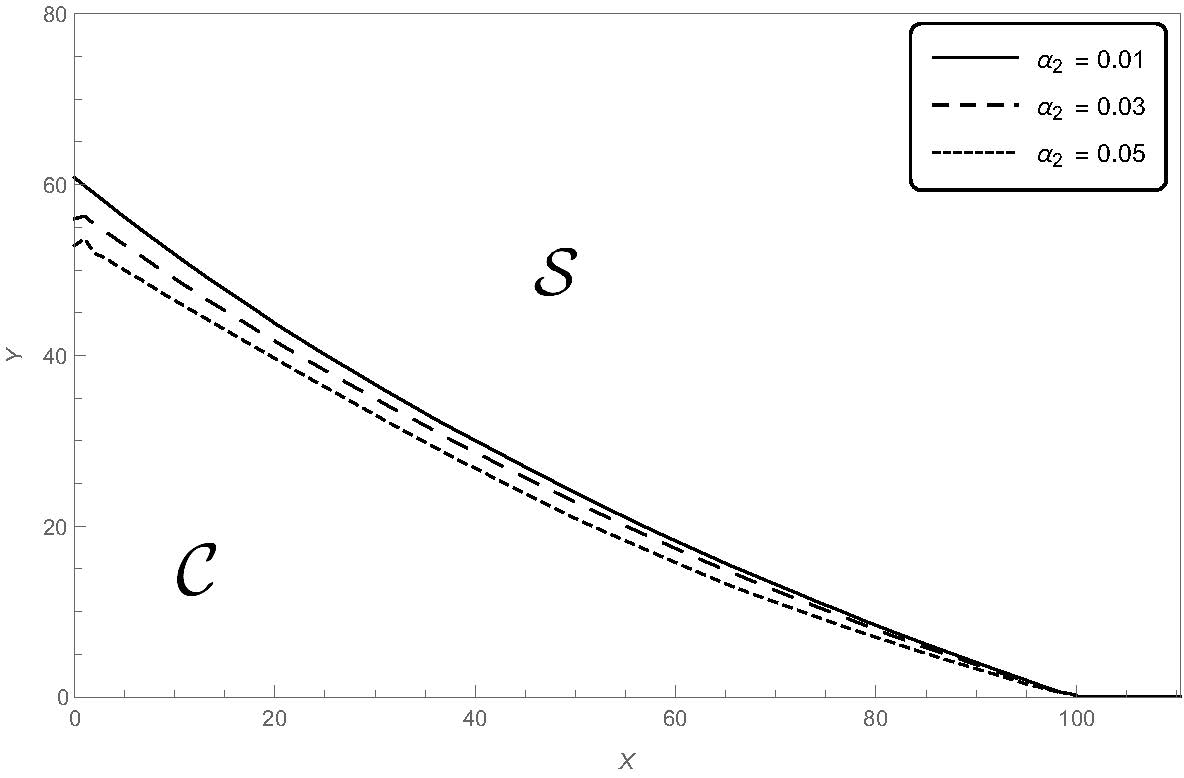}
  \caption{The optimal boundary for different values of $\alpha_2$ and following parameters: $r =0.1 ,\, \alpha_1 = 0.03 ,\,  \sigma_1 = 0.15 ,\, \sigma_2 = 0.2, \, Q_1 =5 ,\, Q_2 = 10 ,\, I = 4000$} \label{Figure: Comp Stat Alpha} 
\end{center}
\end{figure} 
Figure \ref{Figure: Comp Stat Alpha} shows the optimal boundary $b$ for different values of $\alpha_2$. We can see that, differently to what is happening for the volatility, the optimal boundary $b$ is decreasing in $\alpha_2$. A larger drift coefficient $\alpha_2$ implies higher expected prices of the second product on the market and, as a result, the value of the investment increases. To understand the observed effect on the optimal boundary $b$ we notice that the function $F$, which represents the value of exercising the investment option immediately, depends explicitly on $\alpha_2$. Notice that $F$ increases for larger values of $\alpha_2.$ The company thus has an incentive to invest earlier into the production plant and consequently, the boundary decreases. \\
\begin{figure}[h] 
\begin{center}
  \includegraphics[width=9cm]{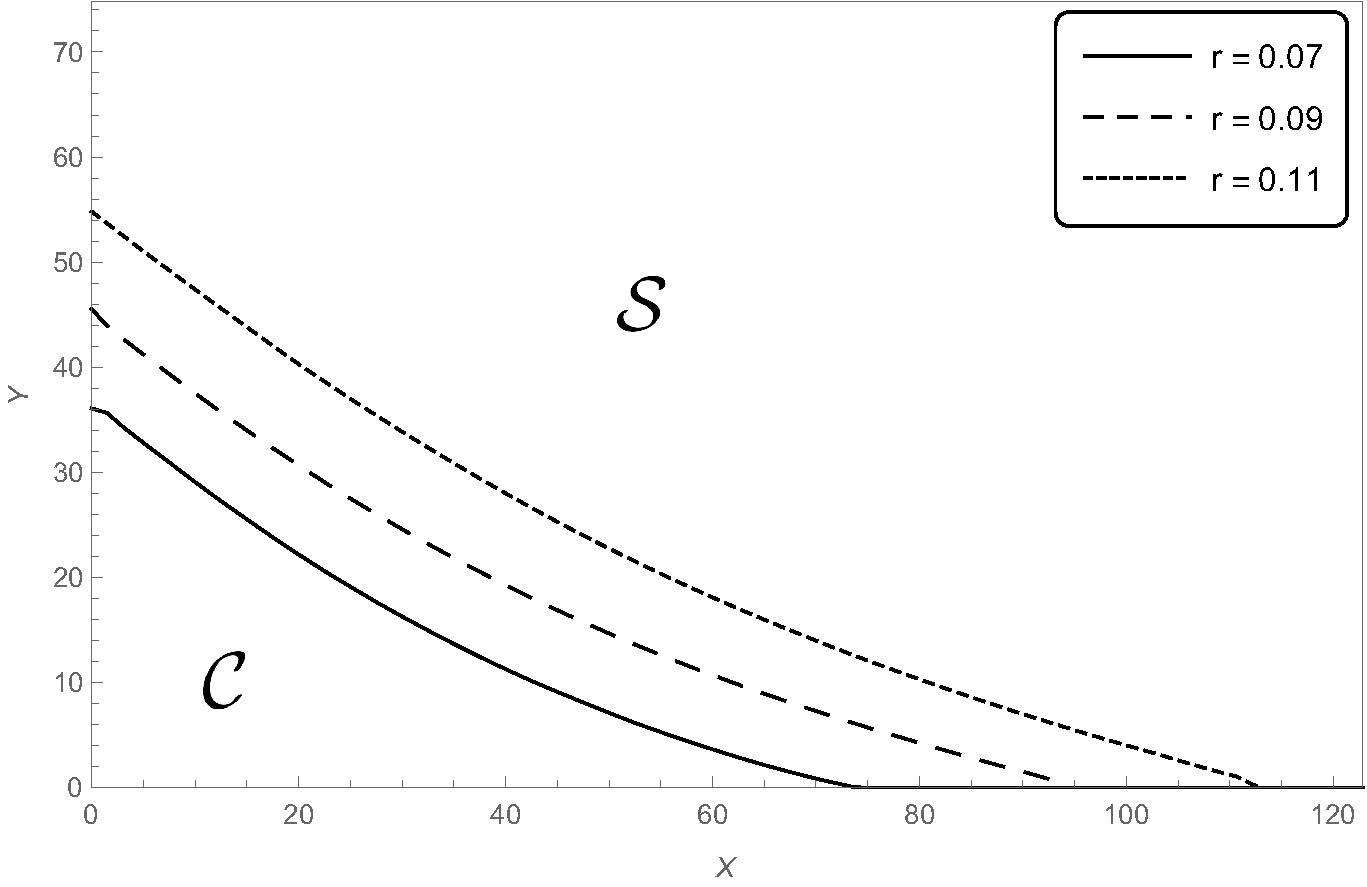}
  \caption{The optimal boundary for different values of $r$ and following parameters: $\alpha_1 = 0.02, \, \alpha_2 = 0.03, \,  \sigma_1 = 0.15 , \, \sigma_2 = 0.15, \, Q_1 =5 ,\, Q_2 = 10 ,\, I = 4000$} \label{Figure: Comp Stat R} 
\end{center} 
\end{figure} \\
In Figure \ref{Figure: Comp Stat R} we can observe the sensitivity of the optimal boundary $b$ with respect to the discount factor $r$. We observe that the boundary $b$ increases in the discount factor. As in the case of a change in the drift coefficient $\alpha_2$, we notice that the function $F$ depends explicitly on $r$. Since $F$ decreases with $r$, the value of exercising the investment immediately decreases, so that the company prefers to delay the investment. Consequently, the boundary $b$ increases. Notice that, differently to what we can observe in Figures \ref{Figure: Comp Stat Sigma} and \ref{Figure: Comp Stat Alpha}, a change in $r$ shifts the investment thresholds on both axes, as in fact $r$ affects both $x^*$ and $y^*$.

\section{Concluding Remarks}
\label{sec:conclusion}

We considered a two-dimensional real options problem of a company facing an irreversible investment decision. By performing a detailed study relying on probabilistic and analytic techniques, we were able to derive a complete characterization of the investment rule. This is triggered by a price-dependent curve that uniquely solves a nonlinear integral equation in a suitable functional class. Such a result has been established by deriving a probabilistic representation of the problem's value function, which, in turn, follows from a suitable approximation procedure whose analysis employs general results from PDE theory. In this respect, the approach that we have used in this work is actually quite flexible and does hold for more general dynamics and payoff functions. Our solution method can be therefore though of as a possible general recipe for the study of multi-dimensional optimal timing problems arising in applications. For instance, real option problems (cf.\ Dixit and Pindyck \cite{Dixit}) in dimension larger than two and with running profits can be considered, as well as financial questions like the so-called exchange-of-baskets-problem (cf.\ Christensen and Salminen \cite{ChrisSalm}, among others), or optimal timing problems in environmental economics (see, e.g., Section 5 in Pindyck \cite{Pindyck-Env}).

Furthermore, differently to other related contributions, we provide a rigorous analytical comparative statics analysis, which is then illustrated through the implementation of a probabilistic numerical scheme for the determination of the optimal investment boundary. Again, the suggested numerical approach does not hinge on the specific problem we are investigating, but it is well suited in order to accommodate other applications as well -- even in larger dimensions.


\appendix
\section{Proof of Proposition \ref{Proposition Properties of value function}} \label{Appendix Proof of properties of value function}
\textit{Lower and Upper Bounds}. Observe that the first lower bound follows by taking the a-priori suboptimal stopping time $\tau = 0$. For the second lower bound, consider the stopping time $\sigma := \sigma (x,y) := \inf \{ t \geq 0 : ~ F(X_{\tau}^{x} , Y_{\tau}^{y} ) > 0 \} $ and notice that $e^{-r \sigma} F(X_{\sigma}^{x} , Y_{\sigma}^{y} ) \one_{ \{ \sigma = \infty\} } = 0$ under the convention (\ref{convention lim e X eins=infty}). It is evident that $\mathbb{P} ( \sigma < + \infty) > 0$, and we thus have 
\begin{align*}
V(x,y) \geq \mathbb{E} [ e^{-r \sigma} F(X_{\sigma}^{x} , Y_{\sigma}^{y}) ] > 0,
\end{align*}
for all $(x,y) \in \mathbb{R}_+$. 
On the other hand, one obtains the upper bound by observing that
\begin{align*}
V(x,y) &= \sup_{\tau \in \mathcal{T}} \mathbb{E} \Bigl[ e^{-r \tau} \Bigl( \frac{Q_{1}}{\delta_{1}} X_{\tau}^{x} + \frac{Q_{2}}{\delta_{2}} Y_{\tau}^{y} - I \Bigl) \Bigl] 
\leq \sup_{\tau \in \mathcal{T}} \Bigl[ \frac{Q_{1}}{\delta_{1}} \mathbb{E} [ e^{-r \tau} X_{\tau}^{x} ] + \frac{Q_{2}}{\delta_{2}}  \mathbb{E} [e^{-r \tau}  Y_{\tau}^{y}] \Bigl] \\
&= \frac{Q_{1}}{\delta_{1}} x + \frac{Q_{2}}{\delta_{2}} y  \leq  C (x +y)
\end{align*} 
upon setting $C := \max \{ \frac{Q_{1}}{\delta_{1}} , \frac{Q_{2}}{\delta_{2}} \} > 0$, and using the uniform integrability stated in Remark \ref{Remark Technicalities}. 
\hspace*{0.2cm}\textit{Monotonicity.} Let $(x,y) \in \mathbb{R}_{+}^{2}$ and $\epsilon > 0$. Consider an $\epsilon$-optimal stopping time $\tau^{\epsilon} := \tau^{\epsilon} (x,y)$ for the optimal investment problem with value function $V(x,y)$. For any $\varphi > 0$, it follows that
\begin{align}
V(x+ \varphi ,y) - V(x,y) + \epsilon &\geq  \mathbb{E}[ e^{-r \tau^{\epsilon}} F(X_{\tau^{\epsilon}}^{x+ \varphi} , Y_{\tau^{\epsilon}}^{y})] - \mathbb{E}[ e^{-r \tau^{\epsilon}} F(X_{\tau^{\epsilon}}^{x} , Y_{\tau^{\epsilon}}^{y})]  
= \frac{Q_{1} \varphi}{\delta_1}\mathbb{E}[ e^{-r \tau^{\epsilon}} X_{\tau^{\epsilon}}^{1}]  \, \geq \, 0, 
\end{align}
where the last inequality holds due to the nonnegativity of $X_{t}^{x}$ and Assumption \ref{Assr>a}. Rearranging terms yields
\begin{align*}
V(x + \varphi, y) + \epsilon \geq V(x,y)
\end{align*}
and $V$ is thus nondecreasing in $x$ by arbitrariness of $\epsilon > 0$. Moreover, by employing similar arguments, we obtain that $V$ is nondecreasing in $y$. \\
\hspace*{0.2cm}\textit{Continuity.} Let $\{ (x_{n} , y_{n}), \,\, n \in \mathbb{N} \} \, \subset \mathbb{R}_{+}^{2}$ be a sequence converging to $(x,y) \in \mathbb{R}_{+}^{2}$. For $\epsilon > 0$, consider an $\epsilon$-optimal stopping time $\tau^{\epsilon} := \tau^{\epsilon} (x,y)$ for the stopping problem with value function $V(x,y)$. It follows that
\begin{align*}
V(x,y) - V(x_{n}, y_{n})  &\leq \epsilon  +  \mathbb{E} [e^{-r \tau^{\epsilon}} ( F(X_{\tau^{\epsilon}}^{x} , Y_{\tau^{\epsilon}}^{y}) \, - \, F(X_{\tau^{\epsilon}}^{x_{n}}, Y_{\tau^{\epsilon}}^{y_{n}}) ) ] \\
&= \epsilon  + (x - x_{n})  \frac{Q_{1}}{\delta_{1}} \mathbb{E} [ e^{-r \tau^{\epsilon}}  X_{\tau^{\epsilon}}^{1} ]  + (y - y_{n}) \frac{Q_{2}}{\delta_{2}}  \mathbb{E} [ e^{-r \tau^{\epsilon}}  Y_{\tau^{\epsilon}}^{1} ].
\end{align*}
and by rearranging terms and letting $n \to \infty$ we obtain 
\begin{align} \label{liminf v(xn,yn) geq v - eps}
\liminf_{n \to \infty} V(x_{n} , y_{n})  \geq  V(x,y) - \epsilon.
\end{align} 
On the other hand, consider an $\epsilon$-optimal stopping time $\tau_{n}^{\epsilon} := \tau^{\epsilon} (x_{n}, y_{n})$ for the stopping problem with value function $V(x_{n}, y_{n})$. 
By noticing that
\begin{align} \label{Expectation e -rt X mit Ito }
\mathbb{E} [ e^{-r \tau} X_{t}^{x} ]  =  x - \mathbb{E} \Bigl[ \int_{0}^{\tau} e^{-rt} \delta_{1} X_{t}^{x} dt \Bigl] 
\quad
\text{and} 
\quad
\mathbb{E} [ e^{-r \tau} Y_{t}^{y} ]  =  y - \mathbb{E} \Bigl[ \int_{0}^{\tau} e^{-rt} \delta_{2} Y_{t}^{y} dt \Bigl].
\end{align}
for all stopping times $\tau \in \mathcal{T}$,  
we obtain
\begin{align*}
V(x_{n}, y_{n}) - V(x,y)  &\leq  \epsilon  + \mathbb{E} [ e^{-r \tau_{n}^{\epsilon}} (F(X_{\tau_{n}^{\epsilon}}^{x_{n}} , Y_{\tau_{n}^{\epsilon}}^{y_{n}} ) - F(X_{\tau_{n}^{\epsilon}}^{x} , Y_{\tau_{n}^{\epsilon}}^{y} ))] 
\\
&=  \epsilon  +  \frac{Q_{1}}{\delta_{1}} \Bigl( (x_{n} - x)  +  \mathbb{E} \Bigl[ \int_{0}^{\tau_{n}^{\epsilon}} e^{-rt} \delta_{1} (X_{t}^{x} - X_{t}^{x_{n}} )dt \Bigl] \Bigl)  \\
& \hspace*{0.67cm} + \frac{Q_{2}}{\delta_{2}} \Bigl( (y_{n} - y)  +  \mathbb{E} \Bigl[ \int_{0}^{\tau_{n}^{\epsilon}} e^{-rt} \delta_{2} (Y_{t}^{y} -  Y_{t}^{y_{n}} ) dt \Bigl] \Bigl) \\
&\leq  \epsilon  +  \frac{Q_{1}}{\delta_{1}} \Bigl((x_{n} -x)  +  \vert x_{n} - x \vert \delta_{1} \mathbb{E} \Bigl[  \int_{0}^{\infty} e^{-rt} X_{t}^{1} dt\Bigl] \Bigl) \\
& \hspace*{0.67cm} +  \frac{Q_{2}}{\delta_{2}} \Bigl( (y_{n} - y)  +  \vert y_{n} - y \vert \delta_{2} \mathbb{E} \Bigl[ \int_{0}^{\infty} e^{-rt} Y_{t}^{1} dt \Bigl] \Bigl)
\end{align*}
and taking the limit as $n \to \infty$ this results to
\begin{align}\label{limsup v leq eps + v}
\limsup_{n \to \infty} V(x_{n} , y_{n})  \leq  \epsilon  +  V(x,y).
\end{align}
The continuity of $V$ then follows from (\ref{liminf v(xn,yn) geq v - eps}) and (\ref{limsup v leq eps + v}) by arbitrariness of $\epsilon > 0$. \\
\hspace*{0.2cm}\textit{Convexity.} Take any  $(x_{1} , y_{1}), (x_{2} , y_{2}) \in \mathbb{R}_{+}^{2}$ and consider a convex combination $(x,y) := \lambda (x_{1} , y_{1}) + (1 - \lambda) (x_{2} , y_{2})$ for $\lambda \in (0,1)$. We obtain 
\begin{align*}
V(x,y) &= \sup_{\tau \in \mathcal{T}} \mathbb{E} \Bigl[ e^{-r \tau} \Bigl( \frac{Q_{1} x X_{\tau}^{1}}{\delta_{1}} + \frac{ Q_{2} y Y_{\tau}^{1}}{\delta_{2}} - I \Bigl) \Bigl] \\
&= \sup_{\tau \in \mathcal{T}} \mathbb{E} \Bigl[ \lambda e^{-r \tau} \Bigl( \frac{Q_{1} x_{1}  X_{\tau}^{1} }{\delta_{1}} + \frac{ Q_{2} y_{1} Y_{\tau}^{1}}{\delta_{2}} - I \Bigl) \\
& \hspace{3.45cm}+ (1 - \lambda) e^{-r \tau} \Bigl( \frac{Q_{1} x_{2}  X_{\tau}^{1} }{\delta_{1}} + \frac{ Q_{2} y_{2} Y_{\tau}^{1}}{\delta_{2}} - I \Bigl) \Bigl] \\[0.15cm]
&\leq \lambda  \sup_{\tau \in \mathcal{T}} \mathbb{E} [ e^{-r \tau} F(X_{\tau}^{x_{1}}, Y_{\tau}^{y_{1}}) ] + (1- \lambda) \sup_{\tau \in \mathcal{T}} \mathbb{E} [ e^{-r \tau} F(X_{\tau}^{x_{2}} , Y_{\tau}^{y_{2}} ) ] \\
&= \lambda V(x_{1} ,y_{1}) + (1 - \lambda) V(x_{2} , y_{2}),
\end{align*}
and the claim follows.
\section{Proof of Theorem \ref{Theorem Probabilistic Repres of v}}\label{Appendix Proof of probabilistic repr.}
We argue by adopting arguments presented in Section 3.1 in De Angelis et al.\,\cite{GF}. At this point it would be convenient for us to study the variational inequality associated with the optimal stopping problem, but since the coefficients of the stochastic differential equations in (\ref{dyn}) are unbounded on the state space $\mathbb{R}_{+}^{2}$, classical results from the PDE literature are not directly applicable. Instead we will approximate the optimal stopping problem (\ref{Value Function}) by a sequence of problems on bounded domains. To this end, define a sequence $\{ Q_{n} , n \in \mathbb{N} \}$ of sets satisfying the following conditions.
\begin{align}
\text{i)}& ~Q_{n} \text{ is open, bounded and connected for all } n \in \mathbb{N} , \label{Qn open, bounded, connected} \\
\text{ii)}& ~ Q_{n} \subset Q_{n+1} \text{ for all } n \in \mathbb{N}, \label{Qn subset Qn+1} \\
\text{iii)}& \lim_{n \to \infty} Q_{n} := \cup_{n \geq 0} Q_{n} = \mathbb{R}_{+}^{2}, \label{lim Qn = R2}\\
\text{iv)}& ~  \partial Q_{n} \in C^{2 + \alpha_{n} } \text{ for some }\alpha_{n} > 0. \label{partial Qn in C 2+alpha}
\end{align}
Note that it is always possible to find such a sequence. Furthermore, we define
\begin{align}\label{Definition sigma n}
\sigma_{n} = \sigma_{n} (x,y) := \inf \{ t \geq 0:~ (X_{t}^{x} , Y_{t}^{y}) \notin Q_{n} \}
\end{align}
and state the following remark. 
\begin{remark}
The condition (\ref{Qn subset Qn+1}) implies that the sequence $\{ \sigma_{n} , n \in \mathbb{N} \}$ is strictly increasing as $n \to \infty$ with limit 
\begin{align}\label{sigma n to sigma infty}
\sigma_{n} \uparrow \sigma_{\infty} = \sigma_{\infty} (x,y) := \inf \{ t \geq 0: ~ (X_{t}^{x}, Y_{t}^{y} ) \notin \mathbb{R}_{+}^{2} \}. 
\end{align}
The boundaries $0$ as well as $+\infty$ of the processes $X_{t}^x$ and $Y_{t}^y$ are natural, meaning they are unattainable whenever the processes are started in the interior of the state space (cf.\,Borodin and Salminen \cite{Borodin}, p.\,136). For the stopping time $\sigma_{\infty}(x,y)$ specified in (\ref{sigma n to sigma infty}) it thus follows that
\begin{align}\label{sigma infty = infty }
\sigma_{\infty} = \sigma_{\infty} (x,y) = \infty \quad \mathbb{P} \text{-a.s.}
\end{align}
for every $(x,y) \in \mathbb{R}_{+}^{2}$. 
\end{remark}
Upon using the stopping time $\sigma_{n}(x,y)$ of (\ref{Definition sigma n}) we localize the optimal stopping problem (\ref{Value Function}) by setting 
\begin{align}\label{Approximating Optimal Stopping problem}
V_{n} (x,y) := \sup_{\tau \in \mathcal{T}} \mathbb{E} [ e^{-r (\tau \wedge \sigma_{n}) } F(X_{ \tau \wedge \sigma_{n} }^{x} , Y_{\tau \wedge \sigma_{n}}^{y} ) ],\quad (x,y) \in \mathbb{R}_{+}^{2}.
\end{align}
 The continuation and stopping regions of this stopping problem are given by
\begin{align}
\mathcal{C}_{n} &:= \{ (x,y) \in \mathbb{R}_{+}^{2}: ~ V_{n} (x,y) > F(x,y) \} 
\hspace*{1cm}
\mathcal{S}_{n} := \{ (x,y) \in \mathbb{R}_{+}^{2}:~ V_{n}(x,y) = F(x,y) \}, \label{Continuation and Stopping region Sn}
\end{align}
respectively. Furthermore, we note that the second-order elliptic differential operator associated with the two-dimensional diffusion $ (X_{t}^{x} , Y_{t}^{y})$ is given by $\mathcal{L} := \mathcal{L}_{X} + \mathcal{L}_{Y}$, where
\begin{align}\label{L infinitesimal generator}
\mathcal{L}_{X} := \frac{1}{2} \sigma_{1}^{2} x^2 \frac{\partial^2}{\partial x^2} + \alpha_{1} x \frac{\partial}{\partial x} \hspace{1cm} \mathcal{L}_{Y} := \frac{1}{2} \sigma_{2}^{2} y^2 \frac{\partial^2}{\partial y^2} + \alpha_{2} y \frac{\partial}{\partial y}, \hspace{2cm}
\end{align}
since we are dealing with two \textit{uncorrelated} geometric Brownian motions (cf.\,Borodin and Salminen \cite{Borodin}, p.\,136). Moreover, by employing standard arguments (cf.\,Peskir and Shiryaev \cite{Peskir}, p.\,49), we can associate the function $V_n \mid_{Q_n}$ to the variational inequality
\begin{align}\label{Variational Inequality for vn}
\max \{ ( \mathcal{L} - r) u(x,y), - u(x,y) + F(x,y) \} = 0, \quad (x,y) \in Q_{n}
\end{align}
with the boundary condition
\begin{align}\label{Boundary condition for vn}
u(x,y) = F(x,y), \quad (x,y) \in \partial Q_{n}.
\end{align}
The following Proposition verifies that the function of (\ref{Approximating Optimal Stopping problem}) indeed solves the system of equations stated in (\ref{Variational Inequality for vn}) and (\ref{Boundary condition for vn}) above.
\begin{proposition}\label{Proposition vn in W2p Qn}
The function $V_{n}$ of (\ref{Approximating Optimal Stopping problem}) uniquely solves the variational inequality (\ref{Variational Inequality for vn}) a.e. in $Q_{n}$ with boundary condition (\ref{Boundary condition for vn}) and we have $V_{n} \in \mathcal{W}^{2,p} (Q_{n})$ for $1 \leq p < \infty$, where $\mathcal{W}^{2,p} (Q_{n})$ denotes the Sobolev space of order 2 (cf.\,Brezis \cite{Brezis}, Chapter 8.2). Furthermore, the stopping time 
\begin{align}\label{Optimal Stopping time tau n}
\tau_{n}^{*} := \inf \{ t \geq 0: ~ (X_{t}^{x}, Y_{t}^{y}) \notin \mathcal{C}_{n} \}, \quad (x,y) \in \mathbb{R}_{+}^{2}
\end{align}
is optimal for the problem (\ref{Approximating Optimal Stopping problem}). 
\end{proposition}
\begin{proof}
$(i)$ The existence and uniqueness of a function $u_{n} \in \mathcal{W}^{2,p}(Q_{n})$ for all $p \in [1, \infty) $ solving the variational inequality (\ref{Variational Inequality for vn}) with boundary condition (\ref{Boundary condition for vn}) is guaranteed by the results derived in Friedman \cite{Friedman}, since the coefficients of the dynamics (\ref{dyn}) are continuous as well as bounded on $\bar{Q}_{n}$ and we have made sufficient assumptions regarding the set $Q_{n}$ and their boundary in (\ref{Qn open, bounded, connected}) and (\ref{partial Qn in C 2+alpha}) (cf.\,Friedman \cite{Friedman}, Theorem 3.2 \& 3.4). Furthermore, we are able to continuously extend the function $u_{n}$ outside of $Q_{n}$ by setting
\begin{align}\label{Extension of un outside Qn}
u_{n} (x,y) = F(x,y), \quad (x,y) \in \mathbb{R}_{+}^{2} \setminus Q_{n}.
\end{align}
In the following we refer to this extension and denote it, with a slight abuse of notation, again by $u_{n}$. \\
\hspace*{0.3cm}
$(ii)$ It is therefore left to check that the value function $V_{n}$ of (\ref{Approximating Optimal Stopping problem}) in fact coincides with the unique solution $u_{n}$ introduced in point $(i)$ over $\mathbb{R}_{+}^{2}$ as well as that the stopping time (\ref{Optimal Stopping time tau n}) is optimal for the problem stated in (\ref{Approximating Optimal Stopping problem}). We first treat the case of  $(x,y) \in \mathbb{R}_{+}^{2} \setminus Q_{n}$, for which the definition of $\sigma_{n}$ in (\ref{Definition sigma n}) evidently yields
$\sigma_{n} (x,y) = 0$. Together with (\ref{Approximating Optimal Stopping problem}) it consequently follows 
\begin{align*}
V_{n} (x,y) = \sup_{\tau \in \mathcal{T} } \mathbb{E} [ e^{-r (\tau \wedge \sigma_{n} )} F(X_{\tau \wedge \sigma_{n}}^{x}, Y_{\tau \wedge \sigma_{n}}^{y}) ] 
= F(x,y). 
\end{align*}
Upon using (\ref{Extension of un outside Qn}), we therefore obtain 
\begin{align*}
V_{n} (x,y) = F(x,y) = u_{n} (x,y), 
\end{align*}
and the claim follows for any $(x,y) \in \mathbb{R}_{+}^{2} \setminus Q_{n}$. 
We now let $(x,y) \in Q_{n}$. In this case we obtain the proof by showing both inequalities $V_{n} \geq u_{n}$ as well as $V_{n} \leq u_{n}$. The theorem of Meyers-Serrin (cf.\,Gilbarg and Trudinger \cite{Gilbarg}, Theorem 7.9) implies there exists a sequence of smooth functions $\{ u_{n}^{k} ( \cdot) ,~  k \in \mathbb{N} \} \subset C^{\infty} (Q_{n})$ such that 
\begin{align}\label{unk converges to un}
u_{n}^{k} \,   \rightarrow \, u_{n}, \quad \text{as} ~ k \to \infty
\end{align}
in $\mathcal{W}^{2,p} (Q_{n})$ for $1 \leq p < \infty$. Since the function $u_{n} $ is continuous and $\bar{Q}_{n}$ is compact, the convergence in (\ref{unk converges to un}) is actually uniform on $\bar{Q}_{n}$ (cf.\,Gilbarg and Trudinger \cite{Gilbarg}, Lemma 7.1). We have enough regularity for the functions $u_{n}^{k} $ to apply Dynkin's formula and obtain 
\begin{align}\label{unk Dynkins formula}
u_{n}^{k} (x,y) = \mathbb{E} \Bigl[ e^{-r (\tau \wedge \sigma_{n}) } u_{n}^{k} (X_{\tau \wedge \sigma_{n} }^{x}, Y_{\tau \wedge \sigma_{n}}^{y} ) - \int_{0}^{\tau \wedge \sigma_{n} } e^{-rt} ( \mathcal{L} - r) u_{n}^{k} (X_{t}^{x} , Y_{t}^{y} ) dt \Bigl]
\end{align}
for any bounded stopping time $\tau$. Using standard localization arguments as well as (\ref{e -r tau f(X,Y) on tau=infty}), one can check that this equality holds true for all stopping times $\tau \in \mathcal{T}$. \\ 
In the next step we will study this equation in the limit as $k \to \infty$. The term on the left-hand side of (\ref{unk Dynkins formula}) converges pointwisely by (\ref{unk converges to un}), whereas the uniform convergence on $\bar{Q}_{n}$ guarantees
\begin{align*}
\lim_{k \to \infty} \mathbb{E} \Bigl[ e^{-r (\tau \wedge \sigma_{n}) } u_{n}^{k} (X_{\tau \wedge \sigma_{n} }^{x}, Y_{\tau \wedge \sigma_{n}}^{y} )\Bigl] \, = \, \mathbb{E} \Bigl[ e^{-r (\tau \wedge \sigma_{n}) } u_{n} (X_{\tau \wedge \sigma_{n} }^{x}, Y_{\tau \wedge \sigma_{n}}^{y} )\Bigl].
\end{align*}
It is therefore left to check that the integral term in (\ref{unk Dynkins formula}) converges, and we need to show 
\begin{align*}
\lim_{k \to \infty} \mathbb{E} \Bigl[ \int_{0}^{\tau \wedge \sigma_{n} } e^{-rt} ( \mathcal{L} - r) u_{n}^{k} (X_{t}^{x} , Y_{t}^{y} ) dt \Bigl]  =  \mathbb{E} \Bigl[ \int_{0}^{\tau \wedge \sigma_{n} } e^{-rt} ( \mathcal{L} - r) u_{n} (X_{t}^{x} , Y_{t}^{y} ) dt \Bigl].
\end{align*} 
We recall Lemma \ref{Lemma Densities} $(ii)$ and take $ q > 1$ suitable for $\bar{Q}_{n}$ and $p$ such that $\frac{1}{q} + \frac{1}{p} =1$. For a multi-index $\alpha$ we specify the following norm on the Sobolev space $\mathcal{W}^{2,p} (Q_{n})$
\begin{align*}
\vert \vert u \vert \vert_{\mathcal{W}^{2,p} (Q_{n})} := \sum_{ \vert \alpha \vert \leq 2 } \vert \vert D^{\alpha} u \vert \vert_{L^{p}(Q_{n})} = \sum_{ \vert \alpha \vert \leq 2 } \Bigl( \int_{Q_{n}} \vert D^{\alpha} u (\xi , \zeta) \vert^{p} d \xi d \zeta \Bigl)^{\frac{1}{p}} 
\end{align*} 
and note that $\mathcal{W}^{2,p} (Q_{n})$ is a Banach space equipped with this norm. 
Hölder's inequality for $p$ and $q$ as defined above then yields 
\begin{align}\label{Expectation Hoelderinequality}
\Bigl\vert \mathbb{E}  \Bigl[ \int_{0}^{\tau  \wedge \sigma_{n} } e^{-rt} ( \mathcal{L} - r)(u_{n}^{k} - u_{n}) (X_{t}^{x} , Y_{t}^{y} ) dt \Bigl] \Bigl\vert 
\leq C \vert \vert u_{n}^{k} - u_{n} \vert \vert_{\mathcal{W}^{2,p}(Q_{n}) },
\end{align}
where $C>0$ denotes a positive constant. The right-hand side of (\ref{Expectation Hoelderinequality}) vanishes as $k \to \infty$ since (\ref{unk converges to un}) holds true and we finally obtain
\begin{align}\label{un Dynkin}
u_{n} (x,y) = \mathbb{E} \Bigl[ e^{-r (\tau \wedge \sigma_{n}) } u_{n} (X_{\tau \wedge \sigma_{n} }^{x}, Y_{\tau \wedge \sigma_{n}}^{y} ) - \int_{0}^{\tau \wedge \sigma_{n} } e^{-rt} ( \mathcal{L} - r) u_{n} (X_{t}^{x} , Y_{t}^{y} ) dt \Bigl], 
\end{align}
for all $\tau \in \mathcal{T}$. As this is well-defined, since $(\mathcal{L} -r) u_{n}$ is defined up to a null set of Lebesgue measure, the variational inequality (\ref{Variational Inequality for vn}) implies on the one hand that
\begin{align*}
u_{n} (x,y) \geq \mathbb{E} [ e^{-r (\tau \wedge \sigma_{n}) } u_{n} (X_{\tau \wedge \sigma_{n} }^{x}, Y_{\tau \wedge \sigma_{n}}^{y} ) ],\quad \forall \tau \in \mathcal{T}
\end{align*}
and furthermore
\begin{align*}
u_{n} (x,y ) \geq \mathbb{E} [ e^{-r (\tau \wedge \sigma_{n})}  F(X_{\tau \wedge \sigma_{n}}^{x} , Y_{\tau \wedge \sigma_{n} }^{y}) ], \quad \forall \tau \in \mathcal{T}.
\end{align*}
By the arbitrariness of $\tau \in \mathcal{T}$ we have 
\begin{align}\label{un geq vn}
u_{n} (x,y) \geq \sup_{\tau \in \mathcal{T}} \mathbb{E} [ e^{-r (\tau \wedge \sigma_{n})} F(X_{\tau \wedge \sigma_{n}}^{x} , Y_{\tau \wedge \sigma_{n} }^{y}) ] 
= v_{n}(x,y)  
\end{align}
and this concludes the first part of the proof. To obtain the reverse, we consider the stopping time 
\begin{align*}
\hat{\tau} = \hat{\tau} (x,y) := \inf \{ t \geq 0: ~ u_{n} (X_{t}^{x}, Y_{t}^{y}) = F(X_{t}^{x}, Y_{t}^{y} ) \}
\end{align*}
and recall $ u_{n} = F$ on $\mathbb{R}_{+}^{2} \setminus {Q}_{n}$. Since $u_{n}$ is continuous on the bounded set $\bar{Q}_{n}$, it is bounded as well and with the convention (\ref{convention lim e X eins=infty}) we have
\begin{align*}
e^{-r (\hat{\tau} \wedge \sigma_{n} )} u_{n} (X_{\hat{\tau} \wedge \sigma_{n} }^{x}, Y_{\hat{\tau} \wedge \sigma_{n} }^{y}) \one_{ \{\hat{\tau} \wedge \sigma_{n} = \infty \} } = \limsup_{t \to \infty} e^{-rt} u_{n} (X_{t}^{x} , Y_{t}^{y}) = 0.
\end{align*}
Consequently, we obtain
\begin{align*}
e^{r (\hat{\tau} \wedge \sigma_{n} ) } u_{n} (X_{\hat{\tau} \wedge \sigma_{n} }^{x}, Y_{\hat{\tau} \wedge \sigma_{n} }^{y} ) &= e^{r (\hat{\tau} \wedge \sigma_{n} ) } u_{n} (X_{\hat{\tau} \wedge \sigma_{n} }^{x}, Y_{\hat{\tau} \wedge \sigma_{n} }^{y} ) \one_{ \{ \hat{\tau} \wedge \sigma_{n} < \infty \} } \\
&= e^{-r (\hat{\tau} \wedge \sigma_{n}) } F(X_{\hat{\tau} \wedge \sigma_{n}}^{x}, Y_{\hat{\tau} \wedge \sigma_{n}}^{y}) \one_{ \{\hat{\tau} \wedge \sigma_{n} < \infty\}} \\
&=  e^{-r (\hat{\tau} \wedge \sigma_{n}) } F(X_{\hat{\tau} \wedge \sigma_{n}}^{x}, Y_{\hat{\tau} \wedge \sigma_{n}}^{y}),
\end{align*}
where the last equality follows from Remark \ref{Remark Technicalities}. Hence, upon using the fact that $( \mathcal{L} -r) u_{n} = 0$ on the set $ \{(x,y) \in Q_{n}: ~ u_{n} (x,y) > F(x,y) \}$ by (\ref{Variational Inequality for vn}), we finally have 
\begin{align}\label{un leq vn}
u_{n} (x,y) &= \mathbb{E} \Bigl[ e^{-r( \hat{\tau} \wedge \sigma_{n}) } u_{n} (X_{ \hat{\tau} \wedge \sigma_{n} }^{x}, Y_{ \hat{\tau} \wedge \sigma_{n} }^{y} ) - \int_{0}^{ \hat{\tau} \wedge \sigma_{n} } e^{-rt} (\mathcal{L} - r)u_{n} (X_{t}^{x} , Y_{t}^{y} )dt \Bigl] \nonumber \\
&= \mathbb{E}[ e^{-r ( \hat{\tau} \wedge \sigma_{n})} F(X_{ \hat{\tau} \wedge \sigma_{n} }^{x} , Y_{ \hat{\tau} \wedge \sigma_{n} }^{y} ) ] 
\leq \sup_{\tau \in \mathcal{T}} \mathbb{E} [ e^{-r (\tau \wedge \sigma_{n} )} F(X_{\tau \wedge \sigma_{n} }^{x}, Y_{\tau \wedge \sigma_{n}}^{y} ) ] \nonumber \\
&= V_{n} (x,y).
\end{align}
Combining (\ref{un geq vn}) and (\ref{un leq vn}) we conclude that $u_{n} = V_{n}$ on $\mathbb{R}_{+}^{2}$. Moreover, as the inequality in (\ref{un leq vn}) becomes an equality, the stopping time $\hat{\tau}$ is optimal for the problem (\ref{Approximating Optimal Stopping problem}) and coincides with the stopping time $\tau_{n}^{*}$ defined in (\ref{Optimal Stopping time tau n}). 
\end{proof}
\begin{remark}\label{Remark unique C1 repr}
In the following, we will refer to the unique $C^{1}$ representative of the elements in the class $\mathcal{W}^{2,p} (Q_{n})$, as the Sobolev inclusions (cf.\,Brezis \cite{Brezis}, Corollary 9.13 \& 9.15) guarantee a continuous embedding of $\mathcal{W}^{2,p} (Q_{n})$ into $C^{1} (\bar{Q}_{n})$ for $p \in (2,\infty)$, and the boundary condition (\ref{Boundary condition for vn}) is thus well posed for such functions. 
\end{remark}
In the next Proposition we derive a probabilistic representation for the value function of (\ref{Approximating Optimal Stopping problem}). The proof follows, apart from a small technicality, by our results stated before in this section. 
\begin{proposition} \label{Proposition Probabilistic Repr of vn}
The function $V_n$ of (\ref{Approximating Optimal Stopping problem}) admits the representation
\begin{align}\label{Probabilistic representation of vn}
V_{n} (x,y) = \mathbb{E} \Bigl[ e^{-r \sigma_{n} } F( X_{\sigma_{n}}^{x} , Y_{\sigma_{n}}^{y} ) - \int_{0}^{\sigma_{n}} e^{-rt} (r I - Q_{1} X_{t}^{x} - Q_{2} Y_{t}^{y}) \one_{ \{ (X_{t}^{x} , Y_{t}^{y}) \in \mathcal{S} \} } dt \Bigl]
\end{align}
for all $(x,y) \in \mathbb{R}_{+}^{2}$. 
\end{proposition}
\begin{proof}
The proof follows by adapting arguments presented in the proof of Proposition \ref{Proposition vn in W2p Qn}, that is finding a sequence (\ref{unk converges to un}), applying Dynkins formula as in  (\ref{unk Dynkins formula}) and taking the limit as $k \to \infty$, we then obtain the representation 
\begin{align}\label{Probabilistic representation of vn in proof}
V_{n} (x,y) = \mathbb{E} \Bigl[ e^{-r \sigma_{n}} F(X_{\sigma_{n}}^{x} , Y_{\sigma_{n}}^{y} ) - \int_{0}^{\sigma_{n}} e^{-rt} (\mathcal{L} - r)V_{n} (X_{t}^{x} , Y_{t}^{y} ) dt \Bigl],
\end{align}
by recalling that $V_{n}$ solves the boundary condition (\ref{Boundary condition for vn}). Moreover, due to Proposition \ref{Proposition vn in W2p Qn} and arguing as in Lemma B.1 in De Angelis et al.\,\cite{GF}, we have  
\begin{align}\label{L -r Vn}
(\mathcal{L} - r) V_{n}(x,y) = (\mathcal{L}- r) F(x,y)\one_{ \{ (x,y) \in \mathcal{S}_{n} \} } = (r I - Q_{1} x - Q_{2} y) \one_{ \{ (x,y) \in \mathcal{S}_{n} \} } 
\end{align}
for a.e. $(x,y) \in Q_{n}$. Due to Lemma \ref{Lemma Densities} we can use (\ref{L -r Vn}) in (\ref{Probabilistic representation of vn in proof}) and the claim follows. 
\end{proof}
In the forthcoming Proposition we explore some properties of the \textit{sequence} of  functions $V_{n}$, most importantly its behaviour in the limit as $n \to \infty$. This is essential for the proof of Theorem \ref{Theorem Probabilistic Repres of v}, as we aim at studying the limit of (\ref{Probabilistic representation of vn}) for $n \to \infty$. 
\begin{proposition}\label{Proposition vn leq v, v conv to v}
The sequence $\{ V_{n} (\cdot) , \,\, n \in \mathbb{N} \}$ is ascending and such that $V_{n}  \leq V $ on $\mathbb{R}_{+}^{2}$ for all $n \in \mathbb{N}$. Moreover, it converges pointwisely to the value function $V$ of the stopping problem (\ref{Value Function}).
\end{proposition}
\begin{proof}
The first two claims follow by recalling (\ref{sigma n to sigma infty}) and simple comparison arguments.
In order to check the convergence of the sequence, we consider an $\epsilon$-optimal stopping time $\tau^{\epsilon} = \tau^{\epsilon} (x,y)$ for the stopping problem with value function  $V(x,y)$. We obtain
\begin{align}\label{0 leq v - vn leq E}
0 \leq V(x,y) - V_{n} (x,y) 
&\leq \mathbb{E} [ e^{-r \tau^{\epsilon} } F(X_{\tau^{\epsilon}}^{x} , Y_{\tau^{\epsilon}}^{y} ) ] - \mathbb{E} [ e^{-r (\tau^{\epsilon} \wedge \sigma_{n}) } F(X_{\tau^{\epsilon} \wedge \sigma_{n} }^{x} , Y_{\tau^{\epsilon} \wedge \sigma_{n}}^{y} ) ]  + \epsilon \nonumber \\[0.3cm]
&= \mathbb{E} [ (e^{-r \tau^{\epsilon}} F(X_{\tau^{\epsilon}}^{x}, Y_{\tau^{\epsilon}}^{y} ) - e^{-r \sigma_{n} } F(X_{\sigma_{n}}^{x} , Y_{\sigma_{n}}^{y} )) \one_{ \{ \sigma_{n} < \tau^{\epsilon} \} } ] + \epsilon.
\end{align}
and due to Assumption \ref{Assr>a} and (\ref{e -rt Xt uniformly integrable}), the sequence of random variables 
\begin{align*}
W_{n} = \Bigl( e^{-r \tau^{\epsilon} } F (X_{\tau^{\epsilon}}^{x} , Y_{\tau^{\epsilon}}^{y} ) - e^{-r \sigma_{n}} F(X_{\sigma_{n}}^{x} , Y_{\sigma_{n}}^{y} ) \Bigl) \one_{ \{ \sigma_{n} < \tau^{\epsilon} \} }
\end{align*}
is uniformly integrable. Moreover, the sequence converges in measure and we have $\lim_{n \to \infty} W_{n} = 0$ $\mathbb{P}$-a.s. The convergence theorem of Vitali (cf.\,Folland \cite{Folland}, p.\,187) then implies
\begin{align*}
\lim_{n \to \infty} \mathbb{E} [ (e^{-r \tau^{\epsilon}} F(X_{\tau^{\epsilon}}^{x}, Y_{\tau^{\epsilon}}^{y} ) - e^{-r \sigma_{n} } F(X_{\sigma_{n}}^{x} , Y_{\sigma_{n}}^{y} )) \one_{ \{ \sigma_{n} < \tau^{\epsilon} \} } ] = 0
\end{align*}
and the claim follows by the arbitrariness of $\epsilon > 0$ in (\ref{0 leq v - vn leq E}). 
\end{proof}
We now state the proof of Theorem \ref{Theorem Probabilistic Repres of v}, in which we derive the probabilistic representation (\ref{Probabilistic representation of v}) of the value function $V$ of the optimal stopping problem (\ref{Value Function}).
\subsubsection*{Proof of Theorem \ref{Theorem Probabilistic Repres of v}}
We study the representation (\ref{Probabilistic representation of vn}) in the limit as $n \to \infty$. Notice that the left-hand side converges pointwisely to the function $V$, due to Proposition \ref{Proposition vn leq v, v conv to v}. 
It is therefore left to check that the following equality holds true
\begin{align} \label{lim E F - int }
\lim_{n \to \infty} \mathbb{E} \Bigl[ e^{-r \sigma_{n} } F( X_{\sigma_{n}}^{x} , Y_{\sigma_{n}}^{y} ) - &\int_{0}^{\sigma_{n}} e^{-rt} (r I - Q_{1} X_{t}^{x} - Q_{2} Y_{t}^{y}) \one_{ \{ (X_{t}^{x} , Y_{t}^{y}) \in \mathcal{S} \} } dt \Bigl] \nonumber \\
= \mathbb{E} \Bigl[ &\int_{0}^{\infty} e^{-rt} (Q_{1} X_{t}^{x} + Q_{2} Y_{t}^{y} - rI) \, \one_{ \{ (X_{t}^{x} , Y_{t}^{y} ) \in \mathcal{S} \}} dt \Bigl].
\end{align}
Note that since 
\begin{align*}
\mathbb{E} [e^{-r \sigma_{n}} F(X_{\sigma_{n}}^{x} , Y_{\sigma_{n}}^{y}) ] = \frac{Q_{1}}{\delta_{1}} \mathbb{E} [e^{-r \sigma_{n}} X_{\sigma_{n}}^{x} ] + \frac{Q_{2}}{\delta_{2}} \mathbb{E} [e^{-r \sigma_{n}} Y_{\sigma_{n}}^{y} ] - e^{-r \sigma_{n}} I
\end{align*}
and $\sigma_{n} \uparrow \infty$, Remark \ref{Remark Technicalities} together with Vitali's convergence theorem yields
\begin{align}\label{lim E F = 0}
\lim_{n \to \infty} \mathbb{E} [e^{-r \sigma_{n}} F(X_{\sigma_{n}}^{x} , Y_{\sigma_{n}}^{y}) ] = 0.
\end{align}
We now seek to study the limit of the integral term of (\ref{lim E F - int }). Observe that $V_{n} \leq V_{n+1} \leq V$ implies $\mathcal{S} \subset \mathcal{S}_{n+1} \subset \mathcal{S}_{n}$ for all $n \in \mathbb{N}$, while the pointwise convergence of $V_n \uparrow V$ implies that $\lim_{n \to \infty} \mathcal{S}_{n} := \bigcap_{n \geq 0} \mathcal{S}_{n} = \mathcal{S}$. We therefore have
\begin{align*}
\lim_{n \to \infty} \one_{ [0, \sigma_{n} ]} (t) e^{-rt} (rI - Q_{1} X_{t}^{x} &- Q_{2} Y_{t}^{y} ) \one_{ \{ (X_{t}^{x} , Y_{t}^{y}) \in \mathcal{S}_{n} \} }  =  e^{-rt} (rI - Q_{1} X_{t}^{x} - Q_{2} Y_{t}^{y} ) \one_{ \{ (X_{t}^{x} , Y_{t}^{y}) \in \mathcal{S} \} }
\end{align*}
for a.e. $(t,x,y) \in \mathbb{R}_{+} \times \mathbb{R}_{+}^{2}$. Moreover, notice that for a constant $C$ depending on $Q_{1}, Q_{2}, I$ and $r$, we have 
\begin{align}\label{vert ert rI - Q1X - Q2Y beschrankt}
\vert e^{-rt} (rI - Q_{1} X_{t}^{x} - Q_{2} Y_{t}^{y} ) \one_{ \{ (X_{t}^{x} , Y_{t}^{y} ) \in \mathcal{S}_{n} \} } \vert  &\leq e^{-rt} \vert rI - Q_{1} X_{t}^{x} - Q_{2} Y_{t}^{y} ) \vert \nonumber \\[0.2cm]
 &\leq e^{-rt} C(1 + X_{t}^{x} + Y_{t}^{y} ),
\end{align}
where the last term is integrable due to Assumption \ref{Assr>a}. Applying dominated convergence then yields
\begin{align}\label{lim E int ert rI - Q1 usw}
\lim_{n \to \infty} \mathbb{E} \Bigl[ \int_{0}^{\sigma_{n}} e^{-rt} (rI - Q_{1} X_{t}^{x} &- Q_{2} Y_{t}^{y} ) \one_{ \{(X_{t}^{x} , Y_{t}^{y}) \in \mathcal{S}_{n} \} } dt \Bigl] \nonumber \\
&= \mathbb{E} \Bigl[ \int_{0}^{\infty} e^{-rt} (rI - Q_{1} X_{t}^{x} - Q_{2} Y_{t}^{y} ) \one_{ \{ (X_{t}^{x} , Y_{t}^{y} ) \in \mathcal{S} \} } dt \Bigl],
\end{align}
and the claim follows by (\ref{lim E F = 0}) and (\ref{lim E int ert rI - Q1 usw}).
\section{Proof of Proposition \ref{Proposition Variational Inequality}} \label{Appendix Lemma Variational Inequality}
By taking the stopping time $\tau = 0$, it immediately follows that $V(x,y) \geq F(x,y)$ for all $(x,y) \in \mathbb{R}_{+}^{2}$.
It is thus left to prove that $(\mathcal{L}- r)V(x,y) \leq 0$ a.e. on $\mathbb{R}_{+}^{2}$. Since the function $F$ is continuous on $\mathbb{R}_{+}^{2}$, standard results from optimal stopping theory (cf.\,Peskir and Shiryaev \cite{Peskir}, Chapter 3, 7.1) and PDE theory for elliptic equations imply that the value function is such that $V \in C^{2,2} (\mathcal{C})$ and it solves 
\begin{align}\label{Dirichlet Problem 1} 
 (\mathcal{L}-r)V = 0 \quad \text{on} ~ \mathcal{C}.
\end{align}
 As $V$ solves \ref{Dirichlet Problem 1}, rearranging terms upon using Proposition \ref{Proposition Properties of value function} implies 
\begin{align*}
0  \leq  \frac{1}{2} \sigma_{1}^{2} x^{2} V_{xx}  &= - \frac{1}{2} \sigma_{2}^{2} y^{2} V_{yy}  - \alpha_{1} x V_{x} - \alpha_{2} y V_{y} + rV 
\leq  r V - \alpha_{1} x V_{x} - \alpha_{2} y V_{y}, \quad \text{on} ~ \mathcal{C},
\end{align*}
which is equivalent to 
\begin{align*}
0  \leq  V_{xx}  \leq  \frac{2}{\sigma_{1}^{2} x^{2}} \Bigl[ rV - \alpha_{1} x V_{x} - \alpha_{2} y V_{y} \Bigl], \quad \text{on} ~ \mathcal{C}. 
\end{align*}
But the right-hand side defines a continuous function on $\mathbb{R}_+^2$ by Proposition \ref{Proposition Value function C1}, and therefore there exists finite 
\begin{align*}
\lim_{\mathcal{C} \ni (x_n , y_n) \to (x,y) \in \partial \mathcal{C} } V_{xx} (x,y).
\end{align*}
 Similarly, one is able to prove existence of the second derivative with respect to $y$ at $\partial \mathcal{C}$. Hence, $V_x$ as well as $V_y$ are therefore locally Lipschitz on $\bar{\mathcal{C}}$, as well as on int$(\mathcal{S})$, where $V=F$. We now show that $V_x$ and $V_y$ are locally Lipschitz continuous on $\mathbb{R}_+^2$. To this end, set $b^{-1} (y) := \inf \{ y \in \mathbb{R}_+ : ~ y > b(x) \}$ and $g(x,y) := - \frac{2}{\sigma_1^2} ( \alpha_1 x V_x (x,y) + \alpha_2 V_y (x,y) + r V(x,y) ) \in C(\mathbb{R}_+^2)$. Then, for $x < b^{-1} (y)$ and $x' > b^{-1} (y)$ we obtain by convexity of $V$ that
\begin{align*}
0 \leq V_x (x' ,y) - V_x (x,y) = \int_x^{b^{-1} (y)} \big( g(z,y) - \frac{\sigma_2^2}{\sigma_1^2} y^2 V_{yy} (z,y) \big) \frac{1}{z^2} dz + \int_{b^{-1}(y)}^{x'} F_{xx} (z,y) dz \\
\leq \int_{x}^{b^{-1} (y)} g (z,y) \frac{1}{z^2} dz \leq K(x,y) (b^{-1} (y) -x ) \leq K(x,y) \vert x' - x \vert ,
\end{align*}
where we used that $V$ solves (\ref{Dirichlet Problem 1}). Here, the constant $K>0$ is such that $K:= K(x,y) \in L^{\infty}_{\text{loc}} (\mathbb{R}^2_+)$. Therefore, $V_x (\cdot, \, y)$ is locally Lipschitz on $\mathbb{R}_+^2$; i.e. $V_{xx}  \in L^{\infty}_{\text{loc}} (\mathbb{R}^2_+)$.  Analogously, one obtains the same result for $y \mapsto V_y (x,y)$. We thus have $V(\cdot,y), V(x,\cdot) \in \mathcal{W}^{2,2}_{\text{loc}} (\mathbb{R}^2_+)$ (cf.\,Evans and Gariepy \cite{Evans}, p.\,164, Theorem 2(ii)), and finally a result by S.\,Bernstein (cf.\,Krantz \cite{Krantz}, Theorem 3) yields $V \in \mathcal{W}^{2,2}_{\text{loc}} (\mathbb{R}^2_+)$. \\
\hspace*{0.3cm} It remains to check that $(\mathcal{L}-r) V \leq 0 $ a.e. in $\mathcal{S}$. For that, we notice  $(\mathcal{L}-r) V  = (\mathcal{L} - r) F $ for a.e. $(x,y) \in \mathcal{S}$, which can in fact be proved by arguing as in De Angelis et al.\,\cite{GF}, Lemma B.1, due to the fact that $V \in \mathcal{W}^{2,2}_{\text{loc}} (\mathbb{R}_+^2)$. Moreover, we have
\begin{align*}
\mathcal{S} = \{ (x,y) \in \mathbb{R}_{+}^{2} : ~ V(x,y) = F(x,y) \}  \subseteq  \{ (x,y) \in \mathbb{R}_{+}^{2} : ~ (\mathcal{L} -r) F(x,y) \leq 0 \}
\end{align*}
which follows from (\ref{h leq b by this sets}). We thus obtain $(\mathcal{L} - r) V(x,y) \leq 0$ for a.e.\ $(x,y) \in \mathcal{S}$, which then completes the proof. 

\section*{Acknowledgements}
The authors thank two anonymous Referees for pertinent and useful comments on a first version of this work. 

\section*{Funding}
The authors gratefully acknowledge funding by the \textit{Deutsche Forschungsgemeinschaft} (DFG, German Research Foundation) – SFB 1283/2 2021 – 317210226.



\begin{thebibliography}{199}
\bibitem{Adkins}
Adkins, R. and Paxson, D.A., Renewing Assets with Uncertain Revenues and Operating Costs. {\itshape Journal of Financial and Quantitative Analysis}, 2011, {\bfseries 46}, 785-813.
 
\bibitem{Alvarez}
Alvarez, L. H.,  Reward functionals, salvage values, and optimal stopping. {\itshape Mathematical Methods of Operations Research}, 2001, {\bfseries 54}(2), 315-337.

\bibitem{Battauz}
Battauz, A., Donno, M.D. and Sbuelz, A., Real Options with a Double Continuation Region. {\itshape Quantitative Finance}, 2011, {\bfseries 12}(3), 465-475. 

\bibitem{Borodin}
Borodin, A.N. and Salminen, P., {\itshape Handbook of Brownian Motion: Facts and Formulae}, 2nd ed., 2002 (Birkh\"{a}user: Basel). 

\bibitem{Borwein}
Borwein, J.M. and Lewis, A.S., {\itshape Convex Analysis and Nonlinear Optimization: Theory and Examples}, 2nd ed., 2006 (Springer: Berlin). 

\bibitem{Brezis}
Brezis, H., {\itshape Functional Analysis, Sobolev Spaces and Partial Differential Equations}, 2011 (Springer: Berlin).
 
 

\bibitem{Cai} Cai, C., De Angelis, T. and Palczewski, J., The American Put with Finite-Time Maturity and Stochastic Interest Rate, 2021. Available online at: \url{https://arxiv.org/pdf/2104.08502.pdf}
 
\bibitem{Christensen2}
Christensen, S. and Irle, A., A Harmonic Function Technique for the Optimal Stopping of Diffusions. {\itshape Stochastics An International Journal of Probability and Stochastic Processes}, 2011, {\bfseries 83}(4-6), 347-363. 

\bibitem{ChrisSalm}
Christensen, S. and Salminen, P., Multidimensional Investment Problem. {\itshape Mathematics and Financial Economics}, 2018, {\bfseries 12}(1), 75-95.
 
\bibitem{Christensen}
Christensen, S., Crocce, F., Mordecki, E. and Salminen, P., On Optimal Stopping of Multidimensional Diffusions. {\itshape Stochastic Processes and their Application}, 2019, {\bfseries 129}(7), 2561-2581. 

\bibitem{Dutch}
Compernolle, T., Huisman, K., Kort, P., Lavrutich, M., Nunes, C. and Thijssen, J.J.J., Investment Decisions with Two-Factor Uncertainty. {\itshape CentER Discussion Paper}, 2018, {\bfseries 2018-003}.

\bibitem{GF}
De Angelis, T., Federico, S. and Ferrari, G., Optimal Boundary Surface with Stochastic Costs. {\itshape Mathematics of Operations Research}, 2017, {\bfseries 42}(4), 1135-1161. 

\bibitem{DeAngelisPeskir}
De Angelis, T., Peskir, G., Global $C^1$ Regularity of the Value Function in Optimal Stopping Problems. {\itshape The Annals of Applied Probability}, 2020, {\bfseries 30}(3), 1007--1031. 

\bibitem{Detemple}
Detemple, J. and Kitapbayev, Y., The Value of Green Energy under Regulation Uncertainty. {\itshape Energy Economics}, 2020, {\bfseries 89}, 104807.

\bibitem{Dixit1}
Dixit, A.K., Entry and Exit Decisions under Uncertainty. {\itshape  Journal of Political Economy }, 1989,  {\bfseries 97}(3),  620-
638.

\bibitem{Dixit}
Dixit, A.K., Pindyck, R.S., {\itshape Investment under Uncertainty}, 1994 (Princeton University Press: Princeton).

\bibitem{Evans}
Evans, L.C. and Gariepy, R.F., {\itshape Measure Theory and Fine Properties of Functions}, 1994 (CRC Press: Boca Raton).

\bibitem{Folland}
Folland, G.B., {\itshape Real Analysis: Modern Techniques and Their Application}, 1999 (John Wiley \& Sons: New York).

\bibitem{Friedman1}
Friedman, A., {\itshape Advanced Calculus}, 1971 (Dover Books on Mathematics: New York).

\bibitem{Friedman}
Friedman, A., {\itshape Variational Principles and Free Boundary Problems}, 1982 (John Wiley \& Sons: New York).

\bibitem{Gerber}
Gerber, H.U. and Shiu, E.S., Martingale Approach to Pricing Perpetual American Options on Two Stocks. {\itshape Mathematical Finance}, 1996, {\bfseries 6}(3), 303-322.

\bibitem{Gilbarg}
Gilbarg, D. and Trudinger, N.S., {\itshape Elliptic Partial Differential Equations of Second Order}, 2001 (Springer: Berlin).

\bibitem{Hu}
Hu, Y. and \O ksendal, B., Optimal Time to Invest when the Price Processes are Geometric Brownian Motions. {\itshape Finance and Stochastics}, 1998, {\bfseries 2}(3), 295-310.

\bibitem{Kara}
Karatzas, I. and Shreve, S.E., {\itshape Brownian Motion and Stochastic Calculus}, 1988 (Springer: Berlin).
\bibitem{Krantz}
Krantz, S.G., An Ontology of Directional Regularity Implying Joint Regularity. {\itshape Real Analysis Exchange}, 2009, {\bfseries 34}(2), 255-266.

\bibitem{Lange}
Lange, R.J., Ralph, D. and St\o re, K., Real-Option Valuation in Multiple Dimensions using Poisson Optional Stopping Times. {\itshape Journal of Financial and Quantitative Analysis}, 2020, {\bfseries 55}(2), 653-677.

\bibitem{Luo}
Luo, P., Xiong, J., Yang, J. and Yang, Z., Real-Options under a Double Exponential Jump-Diffusion Model with Regime Switching and Partial Information. {\itshape Quantitative Finance}, 2019, {\bfseries 19}(6), 1061-1073.

\bibitem{McDonald}
McDonald, R. and Siegel, D., The Value of Waiting to Invest. {\itshape The Quarterly Journal of Economics}, 1986, {\bfseries 101}(4), 707-728.

\bibitem{Myers}
Myers, S.C., Determinants of Corporate Borrowing. {\itshape Journal of Financial Economics}, 1977, {\bfseries 5}(2), 147-176.
\bibitem{Olsen}
Olsen, T.E. and Stensland, G., On Optimal Timing of Investment when Cost Components are Additive and Follow Geometric Diffusions. {\itshape Journal of Economic Dynamics and Control}, 1992, {\bfseries 16}(1), 39-51.
\bibitem{PeskirC}
Peskir, G., Continuity of the Optimal Stopping Boundary for Two-Dimensional Diffusions. {\itshape The Annals of Applied Probability}, 2019, {\bfseries 29}(1), 505-530.
\bibitem{PeskirU}
Peskir, G., On the American Option Problem. {\itshape Mathematical Finance: An International Journal of Mathematics, Statistics and Financial Economics}, 2005, {\bfseries 15}(1), 169-181.
\bibitem{Peskir}
Peskir, G. and Shiryaev, A.N., {\itshape Optimal Stopping and Free-Boundary Problems}, 2006 (Birkh\"{a}user: Basel).
\bibitem{Pindyck1}
Pindyck, R.S., Irreversible Investment, Capacity Choice, and the Value of the Firm. {\itshape American Economic Review}, 1988, {\bfseries 78}, 969-985.
\bibitem{Pindyck2}
Pindyck, R.S.,  Irreversibility, Uncertainty, and Investment. {\itshape Journal of Economic Literature}, 1991, {\bfseries 29}, 1110-1148.

\bibitem{Pindyck-Env}
Pindyck, R.S.,  Irreversibilities and the Timing of Environmental Policy. {\itshape Resource and Energy Economics}, 2000, {\bfseries 22}, 223-259.

\bibitem{Press}
Press, W.H. and Teukolsky, S.A., Fredholm and Volterra Integral Equations of the Second Kind. {\itshape Computers in Physics}, 1990, {\bfseries 4}(5), 554-557. 
\bibitem{Shepp}
Shepp, L. and Shiryaev, A.N., The Russian Option: Reduced Regret. {\itshape The Annals of Applied Probability}, 1993, {\bfseries 3}(3), 631-640.
\bibitem{Stokey}
Stokey, N.L., {\itshape The Economics of Inaction: Stochastic Control Models with Fixed Costs}, 2009 (Princeton University Press: Princeton).
\bibitem{Thijssen}
Thijssen, J.J.J., Irreversible investment and discounting: an arbitrage pricing approach. {\itshape Annals of Finance}, 2010, {\bfseries 6}(3), 295-315.
\bibitem{Trigeorgis}
Trigeorgis, L., {\itshape  Real options in capital investment: Models, strategies, and applications}, 1995 (Greenwood Publishing Group: Westport).
\end{thebibliography}
\end{document}